\documentclass{article}
\usepackage{graphicx}
\usepackage{amsfonts, amsmath, amsthm, amssymb}
\usepackage{booktabs}
\usepackage{xcolor}
\usepackage{mathtools}
\usepackage{setspace}
\usepackage[left=1in,top=1in,right=1in,bottom=1in]{geometry}
\usepackage{multirow}
\usepackage{algorithm}
\usepackage{algorithmic}
\usepackage[round]{natbib}
\usepackage[colorlinks=true, linkcolor=blue, citecolor=blue, urlcolor=blue]{hyperref}

\newcommand{\RR}[0]{\mathbb{R}}

\newtheorem{theorem}{Theorem}

\newtheorem{lemma}{Lemma}

\title{An Estimator-Robust Design for Augmenting Randomized Controlled Trials with External Real-World Data}
\author{Sky Qiu$^1$, Jens Tarp$^2$, Andrew Mertens$^1$, \\ Mark van der Laan$^1$\\ \\
Division of Biostatistics, University of California, Berkeley$^1$\\
Novo Nordisk, Søborg, Denmark$^2$\\
{\tt sky.qiu@berkeley.edu}}
\date{\today}

\begin{document}
\maketitle
\begin{abstract}
Augmenting randomized controlled trials (RCTs) with external real-world data (RWD) has the potential to improve the finite sample efficiency of treatment effect estimators. We describe using adaptive targeted maximum likelihood estimation (A-TMLE) for estimating the average treatment effect (ATE) by decomposing the ATE estimand into two components: a pooled-ATE estimand that combines data from both the RCT and external sources, and a bias estimand that captures the conditional effect of RCT enrollment on the outcome. This approach views the RCT data as the reference and corrects for inconsistencies of any kind between the RCT and the external data source. Given the growing abundance of external RWD from modern electronic health records, determining the optimal strategy to select candidate external patients for data integration remains an open yet critical problem. In this work, we begin by studying the robustness property of the A-TMLE estimator and then propose a matching-based sampling strategy that attempts to improve the robustness of the estimator with respect to the target estimand. Our proposed strategy is outcome-blind and involves matching based on two one-dimensional scores: the trial enrollment score and the propensity score in the external data. We demonstrate in simulations that our sampling strategy improves the coverage and narrows the widths of confidence intervals produced by A-TMLE. We illustrate our method with a case study of augmenting the DEVOTE cardiovascular safety trial by using the Optum Clinformatics claims database.
\end{abstract}

\section{Introduction}
Randomized controlled trials (RCTs) are widely regarded as the regulatory gold standard for establishing treatment efficacy and safety \citep{franklin_nonrandomized_2020}. However, RCTs often face challenges in achieving adequate statistical power to test clinically meaningful hypotheses for regulatory decision making or scientific purposes, especially in cases where outcome events are rare \citep{chen_challenges_2024}, or for testing hypotheses on secondary endpoints or within subgroups the trial was not powered for. For instance, cardiovascular outcome trials for diabetic medications often target rare endpoints, requiring large patient cohorts and lengthy follow-up periods of three to five years \citep{regier_more_2016}. An example is the DEVOTE trial, which was designed to assess the cardiovascular safety of degludec, an ultralong-acting, once-daily basal insulin approved for various patient groups with diabetes \citep{marso_efficacy_2017}. This double-blind, treat-to-target, event-driven trial aimed to rule out the cardiovascular risk of degludec compared to glargine, with a pre-specified noninferiority hazard ratio margin of 1.3 (a 30\% excess risk). While successful in meeting its noninferiority goal, establishing treatment superiority often requires a separate, larger trial with longer follow-up, which presents both logistical and ethical challenges. Forcing patients to remain on a comparator drug for an extended period, especially when there is already suggestive evidence of the investigational treatment's benefits, raises ethical concerns. Such a design delays patient access to a potentially safer treatment option solely to ensure sufficient statistical power for superiority testing. Meanwhile, the growing availability of external real-world data (RWD), including data from electronic health records, medical claims, and clinical registries, offers a promising alternative. Leveraging external RWD to augment RCTs could facilitate testing for superiority in secondary indications and improve safety analyses for rare outcomes, expand the generalizability of findings to understudied populations, and possibly accelerate the timeline for validating superior treatments for regulatory approvals, allowing faster patient access to potentially safer and more effective medications.

Hybrid designs that use external RWD to augment RCT findings present several challenges. In addition to the primary concern of confounding in the external data, there are inconsistencies in outcome measurement between data sources, non-concurrency, variations in the standard-of-care over time, and other differences between trial and routine care. To mitigate some of these differences, researchers can apply the same (or as much as possible, given the available measured covariates in the real-world data source) inclusion and exclusion criteria as the trial when selecting external patients. However, residual heterogeneity between trial and external real-world data may still persist, such as systematic differences in adherence between controlled trial settings and routine clinical care, which can lead to non-comparable exposure definitions. Even when all effect modifiers for trial enrollment are captured, differences in outcome distributions may still arise due to measurement errors. As an example, adverse outcomes are often adjudicated in trials but some marginal events may be missing in real-world databases \citep{kim_when_2020}, and causes of death may not be well-documented in real-world data, creating discrepancies in endpoint definitions between trial and real-world data \citep{dang_case_2023}. Our previous work addressed these issues by proposing an estimator for the ATE for data integration constructed using the adaptive targeted maximum likelihood estimation (A-TMLE) framework \citep{adaptive_tmle_2024,adml_2023}. This approach decomposes the causal estimand into a pooled estimand and a bias estimand. The pooled estimand, although potentially different from the primary target of interest due to various sources of bias, benefits from a larger sample size due to pooling, making it a more efficient target to estimate. The bias estimand then corrects any biases introduced by pooling the external data, treating the RCT as the reference group. This bias correction allows us to avoid unnecessary assumptions beyond the standard RCT identification assumptions. We demonstrated that A-TMLE could provide more precise estimates with narrower confidence intervals while ensuring asymptotically valid confidence interval coverage \citep{adaptive_tmle_2024}.

Given a data integration estimator of choice such as A-TMLE, how investigators should design a hybrid study and enroll external patients remains an important area of interest. For example, selectively including individuals who responded well to treatment could result in an overestimated treatment effect in the target population. To avoid the risk of cherry-picking favorable outcomes, we believe that the sampling of external patients should be conducted in an outcome-blind manner. Additionally, there are natural instances where outcomes may not yet be measured at the design stage. For example, in prospectively designed hybrid trials that collect future outcomes for external patients, financial constraints may require investigators to select a subset for follow-up and outcome measurement. In such cases, obtaining measurements on all available external patients may be infeasible, which further motivates the need for an outcome-blind sampling strategy. While investigators may follow the trial emulation framework \citep{hernan_target_2022} to mimic a trial as closely as possible, it is unclear how specific design choices impact estimator characteristics in data integration settings.

In this article, we study the robustness structure of the A-TMLE estimator and propose a design strategy motivated by this analysis. Specifically, our sampling strategy includes a two-step matching process. First, each RCT patient is matched with $k$ external patients based on the trial enrollment score, defined as the probability of enrolling in the RCT given pre-treatment covariates. Next, we further match based on the propensity score in the external data, defined as the probability of receiving treatment within the selected external patient cohort. This sampling strategy offers three key benefits. First, it improves the performance of A-TMLE in subsequent data integration analyses by controlling an exact remainder term. In ideal scenarios where matching is performed perfectly, this strategy ensures that A-TMLE remains asymptotically unbiased with respect to the target estimand, even when the outcome regression working models used in constructing A-TMLE are misspecified. Second, our strategy stabilizes the estimator’s variance by directly controlling the inversely weighted terms in the efficient influence curves of the target parameters. Third, simulation evidence suggests that this sampling method produces sub-samples with reduced bias. The bias reduction benefits ``test-then-pool" approaches, where efficiency gains are driven by the bias magnitude. For bias-correction methods like A-TMLE, where gains rely on both bias magnitude and complexity, a smaller bias potentially enables the use of simpler, more parsimonious working models to adequately approximate the bias, resulting in more finite-sample efficiency gains. We illustrate our strategy in a case study that augments the DEVOTE trial with external data from Optum's de-identified Clinformatics Data Mart to study the cardiovascular risks of insulin degludec versus insulin glargine.

This article is organized as follows. Section \ref{sec:preliminaries} provides a concise yet self-contained review of the A-TMLE estimation framework and its application to the data integration setting. In Section \ref{sec:design}, we derive the robustness structure of A-TMLE and propose a practical matching-based design inspired by this analysis. In Section \ref{sec:related_work}, we discuss our work within the context of related research, highlighting some connections and distinctions. In Section \ref{sec:simulations}, we present simulation studies that demonstrate the key advantages of our proposed sampling design. In Section \ref{sec:case_study}, we apply our strategy to select external patients from the Optum database to augment the DEVOTE trial. Finally, we conclude with a discussion in Section \ref{sec:discussion}.

\section{Preliminaries}\label{sec:preliminaries}
In this section, we begin with a brief review of the general A-TMLE estimation framework. Next, we define our data integration problem setup and provide an overview of how A-TMLE is applied in this context. Our focus will be on summarizing the key elements necessary to motivate the proposed design strategy. For a more comprehensive discussion of the A-TMLE framework, we refer readers to \citep{adml_2023}, and for its application in the data integration setting, to \citep{adaptive_tmle_2024}. Table \ref{tab:notations} lists the notations used throughout this article for reference.
\begin{table}[]
\centering
\resizebox{\columnwidth}{!}{
\begin{tabular}{c|c}
\toprule
\textbf{Notation} & \textbf{Description} \\ \midrule
$S$ & RCT indicator \\
$W$ & Patient baseline characteristics \\
$A$ & Treatment indicator \\
$Y$ & Outcome \\ \midrule
$Q_P(S,W,A)=E_P(Y\mid S,W,A)$ & Outcome regression \\
$\bar{Q}_P(W,A)=E_P(Y\mid W,A)$ & Outcome regression, marginalized over $S$ \\
$\theta_P(W)=E_P(Y\mid W)$ & Outcome regression, marginalized over $S$ and $A$ \\
$g_P(a\mid W)=P(A=a\mid W)$ & Treatment mechanism \\
$\Pi_P(s\mid W,A)=P(S=s\mid W,A)$ & RCT enrollment mechanism \\
$\tau_{S,P}(W,A)=E_P(Y\mid S=1,W,A)-E_P(Y\mid S=0,W,A)$ & Conditional average RCT-enrollment effect \\
$\tau_{A,P}(W)=E_P(Y\mid W,A=1)-E_P(Y\mid W,A=0)$ & Conditional average treatment effect (CATE) \\ \midrule
$\Psi^F(P_{O,U})=E_W[E(Y_1-Y_0\mid S=1,W)]$ & Covariate-pooled ATE full-data parameter \\
$\Psi^F_2(P_{O,U})=E_W[E(Y_1-Y_0\mid S=1,W)\mid S=1]$ & RCT-only ATE full-data parameter \\
$\tilde{\Psi}(P_0)=E_0[\tau_{A,0}(W)]$ & Pooled-ATE estimand \\
$\tilde{\Psi}_{\mathcal{M}_{A,w}}(P_0)=E_0[\tau_{A,\beta_0}(W)]$ & Pooled-ATE projection estimand \\
$\Psi^\#(P_0)=E_0[\Pi_0(0\mid W,0)\tau_{S,0}(W,0)-\Pi_0(0\mid W,1)\tau_{S,0}(W,1)]$ & Bias estimand \\
$\Psi^\#_{\mathcal{M}_{S,w}}(P_0)=E_0[\Pi_0(0\mid W,0)\tau_{S,\beta_0}(W,0)-\Pi_0(0\mid W,1)\tau_{S,\beta_0}(W,1)]$ & Bias projection estimand \\
\bottomrule
\end{tabular}
}
\caption{Notations and their descriptions.}
\label{tab:notations}
\end{table}

\subsection{Review of A-TMLE}
Adaptive TMLE (also known as ``adaptive debiased machine learning'') is a statistical estimation framework originally proposed to improve the stability of TMLE in challenging scenarios where there is limited overlap in covariates between the treatment and control groups \citep{adml_2023}. A-TMLE can be viewed as adding a layer of regularization to the targeting step within TMLE, resulting in more stable variance estimates in finite samples. For estimating the average treatment effect (ATE) parameter, the A-TMLE procedure involves the following steps. First, highly adaptive lasso (HAL) is used to learn a data-adaptive working submodel for the conditional average treatment effect (CATE) \citep{vdl_generally_2017,benkeser_hal_2016}. This working submodel implies a projection estimand, defined by plugging in the projection of the true data-generating distribution, $P_0$, onto the learned submodel into the original target parameter mapping. TMLE is then constructed for this projection parameter, along with an efficient influence curve-based Wald-type confidence interval. Inference for the projection parameter extends to inference for the nonparametrically defined target parameter due to the difference between these two parameters being second-order under certain conditions \citep{adml_2023}.

Estimators constructed with A-TMLE are generally more efficient in finite samples. This is because if the tangent space of the oracle model (the model that includes the true distribution, $P_0$), denoted by $\mathcal{M}_0$ at $P_0$, is smaller than that of the nonparametric model $\mathcal{M}$, the Cramér-Rao lower bound for the oracle parameter at $P_0$ is also smaller than for the nonparametrically defined target parameter. Consequently, a $P_0$-efficient estimator for $\Psi_{\mathcal{M}_0}$ tends to be $P_0$-super-efficient for $\Psi$. The limiting variance adapts to the size of the oracle model $\mathcal{M}_0$, adjusting to the complexity of $P_0$. A-TMLE uses a specific loss function, commonly referred to as the $R$-loss, to learn the CATE working model \citep{nie_quasi_2021}. The squared-error projection of the outcome regression onto a semiparametric regression working model is equivalent to a weighted squared-error projection of the CATE function onto its corresponding working model. The weights, $g_P(1-g_P)(1\mid W)$, contribute to variance stabilization. This robustness is particularly valuable in challenging scenarios, such as near violations of the positivity assumption, where methods relying on inverse weighting of the treatment mechanism often produce inflated variance estimates.

\subsection{Data integration problem setup}
We observe $n$ independent and identically distributed observations of the random variable $O=(S,W,A,Y)\sim P_0\in\mathcal{M}$, where $P_0$ is the true data-generating distribution and $\mathcal{M}$ is the statistical model. Here, $S\in\{0,1\}$ indicates whether the patient belongs to the RCT; $W\in\RR^d$ represents a set of baseline covariates; $A\in\{0,1\}$ denotes treatment assignment; and $Y\in\RR$ is the clinical outcome of interest. We consider settings where both the RCT and external data have treatment and control arms. For our statistical model $\mathcal{M}$, we assume full knowledge on the treatment mechanism in the RCT, i.e., we assume that $P_0(A=1\mid S=1,W)$ is known. We assume a structural causal model (SCM) \citep{pearl_causality_2009} defined by: $S=f_S(U_S),W=f_W(S,U_W),A=f_A(S,W,U_A)$, and $Y=f_Y(S,W,A,U_Y)$, where $U=(U_S,U_W,U_A,U_Y)$ is a vector of exogenous errors. We define the potential outcomes $Y_1=f_Y(S,W,A=1,U_Y)$ and $Y_0=f_Y(S,W,A=0,U_Y)$ through intervention on $A$. We consider the covariate-pooled ATE target parameter given by
$$
\Psi(P)=E_W[E(Y_1-Y_0\mid S=1,W)],
$$
which represents the treatment effect within the trial, averaged over the pooled covariate distribution. To identify this target parameter as a functional of the observed data distribution, we make the following three assumptions:
\begin{enumerate}
\item[\textbf{A1}] $(Y_0,Y_1)\perp A\mid S=1,W$ (randomization in the RCT); 
\item[\textbf{A2}] $0<P(A=1\mid S=1,W)<1, P_W\text{-a.e.}$ (positivity of treatment assignment in the RCT);
\item[\textbf{A3}] $P(S=1\mid W)>0,P_W\text{-a.e.}$ (positivity of RCT enrollment in the pooled population).
\end{enumerate}
Note that assumptions A1 and A2 are satisfied in an RCT (or a well-designed observational study that captures all the confounders). Although the plausibility of assumption A3 remains debatable at this point, our proposed design in the next section seeks to make this assumption reasonable. Under assumptions A1, A2, and A3, the target estimand is:
$$
\Psi(P_0)=E_0[E_0(Y\mid S=1,W,A=1)-E_0(Y\mid S=1,W,A=0)],
$$
where we remind the readers again that the outer expectation is taken over the covariate distribution of the pooled population. Here, the subscript 0 in $E_0$ means that the expectation is taken under the probability distribution $P_0$. Note that the consistency assumption \citep{rubin_estimating_1974}, which states that $Y_a$ when $A=a$ for $a\in\mathcal{A}$ given $S=1$, is implied by our SCM.

\subsection{Review of A-TMLE for data integration}
In the context of data integration, A-TMLE could be used to gain efficiency when estimating average treatment effects. Applying A-TMLE to this setting involves the following. First, a parameter is proposed that may differ from the primary target parameter of interest but allows full use of the pooled data from RCT and RWD. For data integration, a natural choice is the pooled-ATE estimand, defined as:
$$
\tilde{\Psi}(P_0)=E_0[E_0(Y\mid W,A=1)-E_0(Y\mid W,A=0)].
$$
Note that the identification of $\tilde{\Psi}$ requires the stronger assumption that the treatment effect within the trial can be extended to the pooled population, implying that all the effect modification variables have been captured, which may not hold in practice. As a result, $\tilde{\Psi}$ may induce bias relative to the actual target estimand $\Psi(P_0)$. The induced bias could be defined as:
$$
\Psi^\#(P_0)\equiv \tilde{\Psi}(P_0)-\Psi(P_0).
$$
Our previous work \citep{adaptive_tmle_2024} showed that this bias estimand has an analytic form:
$$
\Psi^\#(P_0)=E_0[\Pi_0(0\mid W,0)\tau_{S,0}(W,0)-\Pi_0(0\mid W,1)\tau_{S,0}(W,1)].
$$
Thus, constructing an A-TMLE for $\Psi(P_0)$ involves applying the A-TMLE estimator twice: once for the pooled-ATE estimand $\tilde{\Psi}(P_0)$ and once for the bias estimand $\Psi^\#(P_0)$. This procedure can also be viewed as producing a bias-corrected estimate of $\tilde{\Psi}(P_0)$. Here, the working model for $\tilde{\Psi}$ is the conditional average treatment effect (CATE), whereas for $\Psi^\#$, it is the conditional average RCT-enrollment effect. These two working models imply a working model $\mathcal{M}_w$ for $\mathcal{M}$, where the subscript $w$ denotes the ``working'' model. 
The projection estimand can then be expressed as:
$$
\Psi_{\mathcal{M}_w}(P_0)=\tilde{\Psi}_{\mathcal{M}_{A,w}}(P_0)-\Psi^\#_{\mathcal{M}_{S,w}}(P_0),
$$
where the subscripts $A$ and $S$ emphasize the respective working models for the conditional average effect of $A$ on $Y$ and $S$ on $Y$.

We chose this particular decomposition of the target estimand in part because it allows us to leverage the potentially simple structure of the bias model. Specifically, the efficiency gain of A-TMLE in this setting, as compared to a nonparametric TMLE, stems from the adaptability in learning the bias working model. It may be reasonable to suspect that the conditional effect of the study indicator $S$ on the outcome $Y$ might be relatively simple or small, given that substantial bias would likely discourage data integration in the first place. Even if the true bias were fully nonparametric, substantial efficiency gains in finite samples are still possible. Additionally, when the bias structure is complex, its magnitude could be small. This is because quality checks of external data sources were often performed to ensure they are high quality and fit-for-purpose, filtering out data sources with large biases. Regardless, it remains essential to reduce heavy dependence on correct specification of the bias working model. Therefore, it is essential to analyze the estimator and understand the conditions under which it remains unbiased, even if the bias working model (and/or the CATE working model) is misspecified.

\section{An Estimator-robust Design}\label{sec:design}
In this section, we investigate the robustness structure of A-TMLE. In the simple point-treatment ATE setting, a regular TMLE is known to be double-robust, meaning that the resulting TMLE estimator will be asymptotically unbiased if either the outcome regression model or the propensity score model is correctly specified \citep{vdl_rose_2011}. This property is particularly desirable in RCTs, where the treatment mechanism is randomized and hence known. Such robustness analysis can be done for our A-TMLE. Specifically, we examine the exact remainder term of the target parameter to identify which components of the likelihood must be correctly specified to ensure asymptotic unbiasedness of the A-TMLE estimator. A natural question is whether this robustness can be exploited to design a sampling scheme for external patients that preserves the unbiasedness of the A-TMLE estimator. We propose a matching-based sampling strategy tailored to achieve this.

\subsection{Robustness properties of A-TMLE for data integration}\label{subsec:robustness}
First, we review the conditions necessary to establish the asymptotic efficiency of a regular estimator. Suppose $\Psi(P_n^\star)$ is an estimator of the target estimand $\Psi(P_0)$. For target parameters that are pathwise differentiable at $P_0$, we have the following expansion:
$$
\Psi(P_n^\star)-\Psi(P_0)=-P_0D^\star(P_n^\star)+R(P_n^\star,P_0),
$$
where $D^\star(P)$ is the efficient influence curve of $\Psi$ at $P$ and $R(P,P_0)$ is the exact remainder \citep{vdl_rose_2011}. If $\Psi(P_n^\star)$ is a TMLE, it solves the score equation $P_nD^\star(P_n^\star)=o_P(n^{-1/2})$. Therefore, we could rewrite the expansion as:
$$
\Psi(P_n^\star)-\Psi(P_0)=(P_n-P_0)D^\star(P_n^\star)+R(P_n^\star,P_0)+o_P(n^{-1/2}).
$$
The term $(P_n-P_0)D^\star(P_n^\star)$ can be controlled by an asymptotic equicontinuity theorem from empirical process theory. Specifically, assume that $D^\star(P_n^\star)$ belongs to a $P_0$-Donsker class with probability tending to one, and if $P_0\{D^\star(P_n^\star)-D^\star(P_0)\}^2$ converges in probability to zero, then $(P_n-P_0)\{D^\star(P_n^\star)-D^\star(P_0)\}=o_P(n^{-1/2})$ \citep{vaart_weak_1997}. As a result, the expansion becomes
$$
\Psi(P_n^\star)-\Psi(P_0)=P_nD^\star(P_0)+R(P_n^\star,P_0)+o_P(n^{-1/2}).
$$
For $\Psi(P_n^\star)$ to be asymptotically linear and efficient, one should be able to express the estimator minus the truth as the empirical mean of the efficient influence curve $D^\star(P_0)$ plus a term that converges in probability to zero at a rate faster than $\sqrt{n}$. Therefore, the critical requirement for asymptotic linearity and efficiency is that the exact remainder $R(P_n^\star,P_0)$ must be $o_P(n^{-1/2})$. To understand the robustness structure of an estimator, it is therefore essential to analyze this exact remainder. Understanding the exact remainder typically involves first finding the canonical gradient of $\Psi$ at $P\in\mathcal{M}$ and computing the exact remainder based on its definition given by
$$
R(P,P_0)\equiv \Psi(P)-\Psi(P_0)+P_0D^\star(P).
$$
Target parameters that admit, for example, double-robust estimation often have an exact remainder that takes the form of an integral of a product of two differences. Each difference is an approximation error that involves the estimate minus its true value for specific components of the likelihood. If one of these components is fully known and correctly specified, making one of the differences zero, the exact remainder becomes zero, thereby guaranteeing asymptotic unbiasedness of the estimator. We follow this strategy to analyze A-TMLE for our target parameter. Lemma \ref{lem:param_unbiased} presents the exact remainders for the pooled-ATE and bias projection estimands.
\begin{lemma}\label{lem:param_unbiased}
The exact remainder $\tilde{R}_{\mathcal{M}_{A,w}}(P,P_0)\equiv \tilde{\Psi}_{\mathcal{M}_{A,w}}(P)-\tilde{\Psi}_{\mathcal{M}_{A,w}}(P_0) + P_0 D_{\mathcal{M}_{A,w}}(P)
$ for the pooled-ATE projection estimand $\tilde{\Psi}_{\mathcal{M}_{A,w}}(P_0)$ is given by
\begin{align*}
\tilde{R}_{\mathcal{M}_{A,w}}(P,P_0)&=\sum_jE_P\phi_j(W)\bigg\{\tilde{I}_{P}^{-1}
P_0\big[(g_P-g_0)(1\mid W)\phi(W)(\theta_P-\theta_0)(W)\\
&+(g_P-g_0)(1-g_P)(1\mid W)\phi(W)(\tau_{A,\beta_P}-\tau_{A,\beta_0})(W)\\
&-(g_P-g_0)^2(1\mid W)\phi(W)\tau_{A,\beta_P}(W)\\
&-(g_P-g_0)(1\mid W)\phi(W)g_0(1\mid W)(\tau_{A,\beta_P}-\tau_{A,\beta_0})(W)\big]\bigg\}_j,
\end{align*}
where $\tilde{I}_P=E_Pg_P(1-g_P)(1\mid W)\phi\phi^\top(W)$. 

The exact remainder $R^\#_{\mathcal{M}_{S,w}}(P,P_0) \equiv \Psi^\#_{\mathcal{M}_{S,w}}(P) - \Psi^\#_{\mathcal{M}_{S,w}}(P_0) + P_0 D_{\mathcal{M}_{S,w}}(P)$ for the bias projection estimand $\Psi^\#_{\mathcal{M}_{S,w}}(P_0)$ is given by
\begin{align*}
R^\#_{\mathcal{M}_{S,w}}(P,P_0)&=P_0\bigg\{\frac{g_P-g_0}{g_P}(1\mid W)\tau_{S,\beta_{P}}(W,1)(\Pi_P-\Pi_0)(1\mid W,1)\\
&-\frac{g_P-g_0}{g_P}(0\mid W)\tau_{S,\beta_P}(W,0)(\Pi_P-\Pi_0)(1\mid W,0)\\
&-(\tau_{S,\beta_P}-\tau_{S,\beta_0})(W,0)(\Pi_P-\Pi_0)(0\mid W,0)\\
&+(\tau_{S,\beta_P}-\tau_{S,\beta_0})(W,1)(\Pi_P-\Pi_0)(0\mid W,1)\\
&+\sum_jE_P[\Pi_P(0\mid W,0)\phi_j(W,0)-\Pi_P(0\mid W,1)\phi_j(W,1)]\bigg[I_{P}^{-1}\cdot\\
&\big[(\Pi_P-\Pi_0)(1\mid W,A)\phi(W,A)(\bar{Q}_P-\bar{Q}_0)(W,A)\\
&+(\Pi_P-\Pi_0)(1-\Pi_P)(1\mid W,A)\phi(W,A)(\tau_{S,\beta_P}-\tau_{S,\beta_0})(W,A)\\
&-(\Pi_P-\Pi_0)^2(1\mid W,A)\phi(W,A)\tau_{S,\beta_P}(W,A)\\
&-(\Pi_P-\Pi_0)(1\mid W,A)\phi(W,A)\Pi_0(1\mid W,A)(\tau_{S,\beta_P}-\tau_{S,\beta_0})(W,A)\big]\bigg]_j\bigg\},
\end{align*}
where $I_P=E_P\Pi_P(1-\Pi_P)(1\mid W,A)\phi\phi^\top(W,A)$.
\end{lemma}
We provide derivations in Appendix \ref{app_R2}. Additionally, note that knowledge of the trial enrollment score $P_0(S=1\mid W)$ and the propensity score in the external data $P_0(A=1\mid S=0,W)$ implies $g_0$ and $\Pi_0$. Specifically, notice that
$$
\Pi(1\mid W,a)=\frac{P(A=a\mid S=1,W)P(S=1\mid W)}{P(A=a\mid S=1,W)P(S=1\mid W)+P(A=a\mid S=0,W)P(S=0\mid W)}.
$$
Due to the treatment assignment in the trial (i.e., the $S=1$-study) being randomized, the probability $P(A=1 \mid S=1, W)$ corresponds to the known randomization probability. Therefore, with additional knowledge on $P(S=1\mid W)$ and $P(A=1 \mid S=0, W)$, one can fully determine $\Pi(1 \mid W, A)$. Moreover, the marginal probability of treatment across both the trial and the external data is given by:
$$
g(1\mid W)=P(A=1\mid S=1,W)P(S=1\mid W)+P(A=1\mid S=0,W)P(S=0\mid W),
$$
which again relies on $P(S=1\mid W)$ and $P(A=1\mid S=0,W)$. Combining these results with Lemma \ref{lem:param_unbiased}, we establish Theorem \ref{thm:robust}, which describes conditions under which the A-TMLE estimator is asymptotically unbiased for the projection estimand $\Psi_{\mathcal{M}_{w}}(P_0)$ and the nonparametrically defined estimand $\Psi(P_0)$, respectively.
\begin{theorem}\label{thm:robust}
If the trial enrollment score $P_0(S=1\mid W)$ and the propensity score in the external data $P_0(A=1\mid S=0,W)$ are correctly specified, then the exact remainder $R_{\mathcal{M}_{w}}(P,P_0)\equiv \Psi_{\mathcal{M}_{w}}(P) - \Psi_{\mathcal{M}_{w}}(P_0) + P_0 D_{\mathcal{M}_{w},P}=0$. Further, if those two scores are constant in $W$, then the oracle bias $\Psi_{\mathcal{M}_w}(P_0)-\Psi(P_0)=0$. In the special case of augmenting with external controls only (i.e., no external treatment arm), $P_0(A=1\mid S=0,W)=0$, which is already a constant. Therefore, in this setting, the only condition required for both the exact remainder for the projection estimand and the oracle bias to be zero is that $P_0(S=1\mid W)$ must also be constant in $W$.
\end{theorem}
\begin{proof}
Given $P_0(S=1\mid W)$ and $P_0(A=1\mid S=0,W)$, we can compute both $g_0(1\mid W)$ and $\Pi_0(1\mid W,A)$. From Lemma \ref{lem:param_unbiased}, we observe that if $g_P=g_0$ and $\Pi_P=\Pi_0$, then the exact remainders $\tilde{R}_{\mathcal{M}_{A,w}}(P,P_0)=0$ and $R^\#_{\mathcal{M}_{S,w}}(P,P_0)=0$. We then show that when both $g_0$ and $\Pi_0$ are constant in $W$, the oracle bias is zero. First, consider the oracle bias for the pooled-ATE estimand, given by  
$$
\tilde{\Psi}_{\mathcal{M}_{A,w}}(P_0) - \tilde{\Psi}(P_0) = E_0[(\tau_{A,\beta_0} - \tau_0)(W)].
$$  
By the definition of the projection parameter,
$$
\beta_{0} = \arg\min_{\beta} E_0 g_0(1 - g_0)(1 \mid W) [(\tau_{A,\beta} - \tau_0)(W)]^2.
$$  
This suggests that $(\tau_{A,\beta_0} - \tau_0)(W)$ is orthogonal to any basis function in $\phi(W)$, including the all-ones vector (intercept term), with respect to an inner product using weights $g_0(1 - g_0)(1 \mid W)$. Since by our assumption, $g_0(1-g_0)(1\mid W)$ is constant in $W$, we have $E_0[(\tau_{A,\beta_0}-\tau_0)(W)]=0$, assuming $g_0(1-g_0)(1\mid W)$ is bounded away from zero. By a similar argument, for $\Psi^\#$, the oracle bias is also zero due to $\Pi_0(0 \mid W, A)$ being constant in $W$. 
\end{proof}
Note that in A-TMLE, we decompose our target estimand into two parts, the pooled-ATE estimand $\tilde{\Psi}(P_0)$ and the bias estimand $\Psi^\#(P_0)$. The estimation of the pooled-ATE estimand relies on a working model $\mathcal{T}_{A,w}=\{\sum_j\beta(j)\phi_{j}(W):\beta\}$ for the CATE function, defined as $\tau_{A,0}(W)=E_0(Y\mid W,A=1)-E_0(Y\mid W,A=0)$. The bias estimand, on the other hand, involves a working model $\mathcal{T}_{S,w}=\{\sum_j\beta(j)\phi_{j}(W,A):\beta\}$ for the conditional average RCT-enrollment effect, $\tau_{S,0}(W,A)=E_0(Y\mid S=1,W,A)-E_0(Y\mid S=0,W,A)$. Theorem \ref{thm:robust} suggests that with correctly specified models for the trial enrollment score and propensity score in the external data, the A-TMLE estimator will be asymptotically unbiased with respect to the projection estimand $\Psi_{\mathcal{M}_w}(P_0)$, even if the two working models $\tau_{A,\beta}$ and $\tau_{S,\beta}$ are misspecified. Furthermore, if these two mechanisms are constant probabilities, the asymptotic unbiasedness of A-TMLE extends to the nonparametrically defined target estimand, which is of primary interest. 

Practically, although it may be challenging to arrange by design so that one has full knowledge on those two factors, having this robustness structure is still beneficial. For instance, consider the exact remainder for the projection parameter $\tilde{\Psi}_{\mathcal{M}_{A,w}}$, which involves the integral of the product of two approximation errors $(g_P-g_0)(1\mid W)(\tau_{A,\beta_P}-\tau_{A,0})(W)$ plus another two approximation errors $(g_P-g_0)(1\mid W)(\theta_P-\theta_0)(W)$. By the Cauchy-Schwarz inequality, the first product can be bounded by the product of $||\bar{g}_P-\bar{g}_0||$ and $||\tau_{A,\beta_P}-\tau_{A,0}||$, where $\bar{g}_P$ denotes the function $g_P(1\mid W)$ and $||\cdot||$ is the $L_2(P)$-norm. Similarly, for the other product of two differences $(g_P-g_0)(1\mid W)(\theta_P-\theta_0)(W)$, one could bound it by the product of $||\bar{g}_P-\bar{g}_0||$ and $||\theta_P-\theta_0||$. To establish asymptotic efficiency, the exact remainder must satisfy $o_P(n^{-1/2})$. This implies that each difference in the product needs to be $o_P(n^{-1/4})$. However, if one component can be estimated faster than $n^{-1/4}$, it allows for a slower rate for the other component, as long as their product remains $o_P(n^{-1/2})$. This property is particularly appealing when, for example, $g_0(1\mid W)$ can be estimated at a rate faster than $n^{-1/4}$. In such cases, a slower rate for estimating the CATE function $\tau_{A,0}(W)$ and $\theta_P(W)$ is permitted, as long as the product of their errors remains $o_P(n^{-1/2})$. The projection parameter $\Psi^\#_{\mathcal{M}_{S,w}}$ has a similar story, if $\Pi_0(1\mid W,A)$ can be estimated at a rate faster than $n^{-1/4}$, then the conditional average RCT-enrollment effect $\tau_{S,0}(W,A)$ and $\bar{Q}_0(W,A)$ could be estimated at slower rates than $n^{-1/4}$, provided that the product of their errors is $o_P(n^{-1/2})$. It is also important to note that these two scores are outcome-blind, allowing us to control them without looking at outcomes.

To summarize, in the most ideal scenario, both the trial enrollment score and the propensity score in the external data are fully known and constant in $W$. Under these conditions, even if the working models for the CATE function and the conditional average RCT-enrollment effect are misspecified, A-TMLE remains asymptotically unbiased for the nonparametrically defined target parameter. This is due to the exact remainder term being zero. In scenarios where the functional forms of the trial enrollment score and the propensity score are fully known but depend on $W$, A-TMLE is still guaranteed to be asymptotically unbiased with respect to the projection parameter. This parameter is defined as plugging in the projection of the true data-generating distribution onto a data-adaptively learned working model, such as one constructed using HAL. As we will demonstrate later in this section, numerical evaluations of the oracle bias, defined as the difference between the projection parameter and the nonparametrically defined target parameter, suggest that the oracle bias tends to be small in finite samples under various degrees of model misspecification. In the least ideal scenario, neither the trial enrollment score nor the propensity score is known to the investigator. This is the challenge we aim to address with our proposed design strategy, introduced in the next subsection. Specifically, we propose a method for selecting external patients in a way to improve control over the trial enrollment and propensity scores in the sampled subset of external data. Faster convergence of these scores would allow for slower convergence of the working models while still ensuring that the exact remainder term is $o_P(n^{-1/2})$.

\subsection{Robustness structure inspired sampling of external patients}\label{subsec:sampling}
While it may not always be possible to have pre-existing knowledge on the trial enrollment score in the pooled data and the propensity score in the external data, it may be feasible to come up with a sampling strategy such that a subsample of external patients, when combined with the RCT, would make these scores constants (or, at least, close to constants), so that one could exploit the robustness structure of A-TMLE to obtain a more robust estimator. Specifically, we propose starting with an initial filtering of external patients based on the trial’s inclusion and exclusion criteria. This step ensures that the identification assumption A3, which requires that each individual in the pooled population has a non-zero probability of trial participation, is made more plausible. After this step, some dependence of trial enrollment on patient baseline covariates may still remain, if, for example, the RCT sample and RWD sample are not drawn i.i.d from the trial-eligible population. To further address this, we propose matching each subject in the RCT with $k>1$ external real-world subjects.  Matching is a widely used technique in observational studies to balance covariate distributions between groups, thereby reducing bias from covariates \citep{rosenbaum_ps_1983, stuart_matching_2010}. Here, the goal of matching is to arrange so that $S$ is approximately independent of $W$ within the matched sample of RCT and external RWD, that is, $P(S=1\mid W)\approx 1/(1+k)$. If only external control subjects are available, this step concludes the matching process. However, if augmentation with external treated subjects is also desired, we propose an additional step that further matches each selected external treated patient with $m>0$ external control patients within the previously matched cohort. This ensures that within the selected external cohort, treatment assignment is roughly independent of covariates with $P(A=1\mid S=0,W)\approx 1/(1+m)$. 

Ideally, exact matching on all baseline covariates would achieve the desired objectives. In particular, if both treatment and trial enrollment are ignorable given the observed baseline covariates, and the sample size of external RWD is large enough to find perfect matches, the trial enrollment score in the pooled sample and the propensity score within the selected external cohort could be made constant in $W$, guaranteeing the asymptotic unbiasedness of A-TMLE with respect to the nonparametrically defined target estimand. However, exact matching is often impractical when covariates are high-dimensional or continuous. In such cases, alternative distance measures for matching can be adopted. For instance, one may perform propensity score matching. Specifically, one may first match each RCT participant with $k>1$ external patients using the one-dimensional trial enrollment score $P(S=1\mid W)$. After that, the propensity score within the external subset (obtained from trial enrollment score matching) $P(A=1\mid S=0,W)$ is estimated and used as a one-dimensional score for propensity score matching. Even if the trial enrollment score or propensity score models are misspecified and/or matches are inexact, there remains the possibility that the trial enrollment and propensity score post-matching fall within a smaller class of functions. This could potentially allow for faster-than-$n^{-1/4}$ rates of estimation for the trial enrollment and propensity score, allowing for possibly slower rates of estimation for the CATE and conditional average RCT-enrollment effect.

We make the following remarks regarding applying the matching design to select external data in practical settings. First, in certain practical scenarios, there may be a much greater number of external controls available compared to treated patients. This situation often arises when the drug has been recently approved at the time of data integration analysis, leading to limited real-world exposure to treatment. In such cases, we recommend in the propensity score matching step to match each treated patient with multiple external controls to maximize the external cohort's sample size for augmentation. In our case study, we indeed observed a smaller number of treated individuals in the RWD compared to controls. This imbalance was due to the active treatment, insulin degludec, being on the market for a shorter time than the comparator drug. We therefore chose $m=5$, that is, we used 1:5 matching during the propensity score matching step, ensuring that the external sample size was not constrained by the relatively limited number of external treated individuals.

Second, the interpretation of the final estimate should be considered when selecting individuals for sampling. Different matching specifications, such as the choice of distance measure, can result in varied patient profiles in the pooled sample. Researchers should ensure that the resulting population is of clinical relevance. Specifically, note that our target estimand $\Psi(P_0)$ takes the expectation over the pooled covariate distribution. Therefore, the interpretation applies to a pooled RCT and external RWD population, assumed to be drawn independently and identically distributed. One also has the option to modify the target estimand to the sample average treatment effect, so the interpretation is limited to the selected pooled sample. This may also improve the efficiency of A-TMLE further by removing the contribution from the covariate distribution in the variance estimate.

Third, in cases where certain patient groups are underrepresented in the RCT, incorporating external data may help the pooled sample better reflect the patient distributions in the target population. In these situations where the goal is to improve the generalizability of RCT findings, it may be preferable not to set the trial enrollment probability independent of $W$, as doing so would likely result in a sample that still resembles a skewed RCT population. Instead, one could sample external patients in a way so that the trial enrollment score remains a known function of $W$ rather than a constant, so that one still preserves asymptotic unbiasedness with respect to the projection parameter. Robustness with respect to the projection parameter itself is valuable for several reasons. First, in much of the existing literature, target parameters are often defined based on a pre-specified parametric working model. These parametric models are subject to bias introduced by model misspecifications, which will also reduce the interpretability of the final estimate. In contrast, our projection parameter is defined on a data-adaptively learned submodel, which would grow with sample size. Specifically, our working model could be learned using HAL, a flexible nonparametric regression method able to approximate any function with a bounded variation norm, which includes most functions encountered in health domains. Target parameters defined through such data-adaptive working models significantly reduce bias due to model misspecification compared with pre-specified parametric models. Second, even when the model is misspecified due to unobserved variables or incorrect functional forms, the resulting oracle bias in practice is often small due to many potential cancellations. To illustrate, we conducted numerical evaluations assessing the magnitude of oracle bias across various types and degrees of model misspecification, with results presented in Table \ref{tab:oracle_bias}. Overall, we observed that oracle bias tends to be small relative to the true effect size (0.5 in this case). If the oracle bias is of concern, we recommend conducting sensitivity analyses, such as the numerical evaluations presented here, to assess the magnitude of the oracle bias relative to the estimated effect size.
\begin{table}[h!]
\centering
\resizebox{\columnwidth}{!}{
  \begin{tabular}{|cll|rll|lll|}
  \toprule
  \multicolumn{3}{c|}{\textbf{Unobserved}} & \multicolumn{3}{c|}{\textbf{Misspecified}} & \multicolumn{3}{c}{\textbf{Unobserved + Misspecified}} \\ \midrule
  \multicolumn{1}{c}{Adjusted $R^2$} &
    \multicolumn{1}{c}{Weighted MSE} &
    \multicolumn{1}{c|}{Oracle bias} &
    \multicolumn{1}{c}{Adjusted $R^2$} &
    \multicolumn{1}{c}{Weighted MSE} &
    \multicolumn{1}{c|}{Oracle bias} &
    \multicolumn{1}{c}{Adjusted $R^2$} &
    \multicolumn{1}{c}{Weighted MSE} &
    \multicolumn{1}{c}{Oracle bias} \\ \midrule
  \multicolumn{1}{c}{0.92} & \multicolumn{1}{c}{$1.07 \times 10^{-3}$} & \multicolumn{1}{c|}{$-5.68 \times 10^{-5}$} & \multicolumn{1}{c}{0.95} & \multicolumn{1}{c}{$1.10 \times 10^{-4}$} & \multicolumn{1}{l|}{$-1.52 \times 10^{-6}$} & \multicolumn{1}{c}{0.92} & \multicolumn{1}{c}{$1.17 \times 10^{-3}$} & \multicolumn{1}{c}{$6.42 \times 10^{-5}$} \\ 
  \multicolumn{1}{c}{0.75} & \multicolumn{1}{c}{$4.29 \times 10^{-3}$} & \multicolumn{1}{c|}{$-1.27 \times 10^{-4}$} & \multicolumn{1}{c}{0.84} & \multicolumn{1}{c}{$4.43 \times 10^{-4}$} & \multicolumn{1}{l|}{$-1.62 \times 10^{-5}$} & \multicolumn{1}{c}{0.73} & \multicolumn{1}{c}{$4.72 \times 10^{-3}$} & \multicolumn{1}{c}{$1.41 \times 10^{-4}$} \\ 
  \multicolumn{1}{c}{0.57} & \multicolumn{1}{c}{$9.62 \times 10^{-3}$} & \multicolumn{1}{c|}{$6.85 \times 10^{-5}$} & \multicolumn{1}{c}{0.73} & \multicolumn{1}{c}{$9.75 \times 10^{-4}$} & \multicolumn{1}{l|}{$-8.38 \times 10^{-5}$} & \multicolumn{1}{c}{0.55} & \multicolumn{1}{c}{$1.05 \times 10^{-2}$} & \multicolumn{1}{c}{$3.18 \times 10^{-4}$} \\ 
  \multicolumn{1}{c}{0.43} & \multicolumn{1}{c}{$1.72 \times 10^{-2}$} & \multicolumn{1}{c|}{$-1.97 \times 10^{-4}$} & \multicolumn{1}{c}{0.63} & \multicolumn{1}{c}{$1.76 \times 10^{-3}$} & \multicolumn{1}{l|}{$1.48 \times 10^{-5}$} & \multicolumn{1}{c}{0.41} & \multicolumn{1}{c}{$1.87 \times 10^{-2}$} & \multicolumn{1}{c}{$4.04 \times 10^{-4}$} \\ 
  \multicolumn{1}{c}{0.32} & \multicolumn{1}{c}{$2.69 \times 10^{-2}$} & \multicolumn{1}{c|}{$-7.35 \times 10^{-5}$} & \multicolumn{1}{c}{0.56} & \multicolumn{1}{c}{$2.80 \times 10^{-3}$} & \multicolumn{1}{l|}{$1.09 \times 10^{-4}$} & \multicolumn{1}{c}{0.32} & \multicolumn{1}{c}{$2.94 \times 10^{-2}$} & \multicolumn{1}{c}{$2.23 \times 10^{-4}$} \\ 
  \multicolumn{1}{c}{0.25} & \multicolumn{1}{c}{$3.86 \times 10^{-2}$} & \multicolumn{1}{c|}{$-2.03 \times 10^{-6}$} & \multicolumn{1}{c}{0.51} & \multicolumn{1}{c}{$3.91 \times 10^{-3}$} & \multicolumn{1}{l|}{$1.69 \times 10^{-4}$} & \multicolumn{1}{c}{0.25} & \multicolumn{1}{c}{$4.23 \times 10^{-2}$} & \multicolumn{1}{c}{$2.47 \times 10^{-4}$} \\ 
  \multicolumn{1}{c}{0.20} & \multicolumn{1}{c}{$5.23 \times 10^{-2}$} & \multicolumn{1}{c|}{$-7.16 \times 10^{-5}$} & \multicolumn{1}{c}{0.47} & \multicolumn{1}{c}{$5.48 \times 10^{-3}$} & \multicolumn{1}{l|}{$1.00 \times 10^{-4}$} & \multicolumn{1}{c}{0.20} & \multicolumn{1}{c}{$5.75 \times 10^{-2}$} & \multicolumn{1}{c}{$5.11 \times 10^{-4}$} \\ 
  \multicolumn{1}{c}{0.16} & \multicolumn{1}{c}{$6.85 \times 10^{-2}$} & \multicolumn{1}{c|}{$-5.71 \times 10^{-4}$} & \multicolumn{1}{c}{0.44} & \multicolumn{1}{c}{$7.04 \times 10^{-3}$} & \multicolumn{1}{l|}{$2.27 \times 10^{-4}$} & \multicolumn{1}{c}{0.17} & \multicolumn{1}{c}{$7.52 \times 10^{-2}$} & \multicolumn{1}{c}{$5.86 \times 10^{-4}$} \\ 
  \multicolumn{1}{c}{0.13} & \multicolumn{1}{c}{$8.70 \times 10^{-2}$} & \multicolumn{1}{c|}{$4.06 \times 10^{-4}$} & \multicolumn{1}{c}{0.42} & \multicolumn{1}{c}{$8.96 \times 10^{-3}$} & \multicolumn{1}{l|}{$2.25 \times 10^{-6}$} & \multicolumn{1}{c}{0.15} & \multicolumn{1}{c}{$9.51 \times 10^{-2}$} & \multicolumn{1}{c}{$6.71 \times 10^{-4}$} \\ 
  \multicolumn{1}{c}{0.11} & \multicolumn{1}{c}{$1.07 \times 10^{-1}$} & \multicolumn{1}{c|}{$8.23 \times 10^{-4}$} & \multicolumn{1}{c}{0.40} & \multicolumn{1}{c}{$1.09 \times 10^{-2}$} & \multicolumn{1}{l|}{$9.86 \times 10^{-7}$} & \multicolumn{1}{c}{0.13} & \multicolumn{1}{c}{$1.17 \times 10^{-1}$} & \multicolumn{1}{c}{$1.01 \times 10^{-3}$} \\ 
  \bottomrule
  \end{tabular}
}
\caption{Numerical evaluations of the oracle bias, i.e., difference between the projection estimand and the nonparametrically defined estimand under various types and degrees of misspecification of the bias working model.}
\label{tab:oracle_bias}
\end{table}

Finally, note that our proposed design is not limited to the setting where one wishes to use A-TMLE as the estimator of choice in the downstream data integration analysis. In particular, weakening the dependence of trial enrollment on patient baseline characteristics also helps with stabilizing the efficient influence function of the target parameter $\Psi$ in the nonparametric model. To see that, we previously derived in \citep{adaptive_tmle_2024} that the efficient influence function of the target parameter $\Psi$ at $P$ is given by
\begin{align*}
D_{\Psi,P}(O)&=Q_P(1,W,1)-Q_P(1,W,0)-\Psi(P)\\
&+\frac{S}{P(S=1\mid W)}\cdot\frac{2A-1}{P(A\mid 1,W)}(Y-Q_P(S,W,A)).
\end{align*}
In particular, the nonparametric EIC involves inverse weighting of the trial enrollment score. Therefore, our design may also help in obtaining a more stable variance estimate if we were to estimate the target parameter nonparametrically using efficient estimators, such as a standard TMLE.

\section{Related Work}\label{sec:related_work}
From our literature review, we identified three settings for augmenting RCTs with external data. The first setting is where the RCT contains both treatment and control arms, and the external data also includes both treatment and control groups \citep{zhang_estimation_2021,li_augmenting_2022}. The second is where the RCT contains both treatment and control arms, but the external data includes only controls \citep{stuart_matching_2008,lin_propensity_2018,yuan_design_2019,liu_matching_2022,li_improved_2023,lin_matching_2023}. The third setting is a single-arm RCT with no concurrent controls, relying solely on external controls \citep{wang_propensity_2020,li_external_2021}. Single-arm trials are arguably the more challenging ones, as it often requires additional untestable assumptions to ensure valid inference for the causal estimand due to the absence of concurrent controls. 

Various matching-based designs have been proposed that are tailored to these different settings. In \citep{zhang_estimation_2021}, the authors proposed matching to balance the covariate distributions between the RCT and RWD, aligning with the first step of our matching design, which uses the trial enrollment score. They have demonstrated that such matching might help reduce the bias. However, their approach does not involve the additional propensity score matching within the external cohort in the second step of our strategy. Their method produces both the pooled estimate and the RCT-only estimate, leaving it to the reader to interpret any differences between the two. The setting considered in \citep{stuart_matching_2008} falls under the second setting. Although their work focused on an observational setting, they assumed that treatment assignment was ignorable in the combined treatment and primary control group. In other words, they assumed that all confounders were captured in the primary study. Statistically, this is similar (essentially equivalent) to the setting of an RCT, where randomization eliminates backdoor paths from treatment to covariates. A key difference between their setting and RCT and RWD data integration is that they attempted to only use a subset of primary controls in the final analysis because, in their observational setting, there was limited overlap between the treatment and primary control groups. Variations of this approach have since been adopted in the context of RCTs augmented with external controls. For example, \citep{yuan_design_2019} adapted the method to settings with an underpowered concurrent control arm, using external control patients to improve balance and ensure the treatment and control arms are of equal size. Later, \citep{li_improved_2023} proposed an improvement to this approach. Instead of matching based solely on a subset of treated patients from the treatment arm (with a sample size equal to the number of required matches from the external data), they proposed matching the entire concurrent RCT. This approach is similar to ours when there are only controls in the RWD. However, in our setting, we match based on the trial enrollment score rather than the propensity score. Interestingly, these two scores are proportional, making them theoretically equivalent in this context. To see this equivalence, consider a scenario where the external data contains only controls. Here, $A=1$ implies $S=1$, meaning that treated patients must belong to the RCT (as no treated patients exist in the external data). Thus, the propensity score can be expressed as: $P_0(A=1\mid W)=P_0(A=1,S=1\mid W)=P_0(S=1\mid W)P_0(A=1\mid S=1,W)$. Given that treatment in the RCT is randomized, say with probability 0.5, we have $P_0(A=1\mid S=1,W)=0.5$. Therefore, the propensity score in this setting is simply half of the trial enrollment score. It is important to note that the estimand may change depending on the matching strategy used. For a detailed discussion on how different matching strategies target different estimands in the context of RCT and RWD integration, we refer readers to \citep{lin_matching_2023}.

\section{Simulations}\label{sec:simulations}
We designed simulations to empirically evaluate our proposed matching-based design strategy. Specifically, we considered a scenario where investigators have access to an RCT with a fixed sample size of $n=400$ participants and a large pool of $n=25,000$ external patients from five different data sources. Each external source introduces a slightly different functional form and magnitude of bias, ordered from least biased (source 1) to most biased (source 5). Additionally, the baseline covariates from each source have varying degrees of covariate shift relative to the RCT covariate distribution, with source 5 having the largest covariate shifts. Details of the data-generating process are provided in Appendix B. This simulation setup reflects realistic scenarios where external data sources might come from different regions, hospital facilities, or time periods. For simplicity, we assume that all the external patients are trial-eligible, although in practice, one might need to filter patients based on inclusion and exclusion criteria first. 
\begin{figure}[h!]
\centering
\includegraphics[width=1\linewidth]{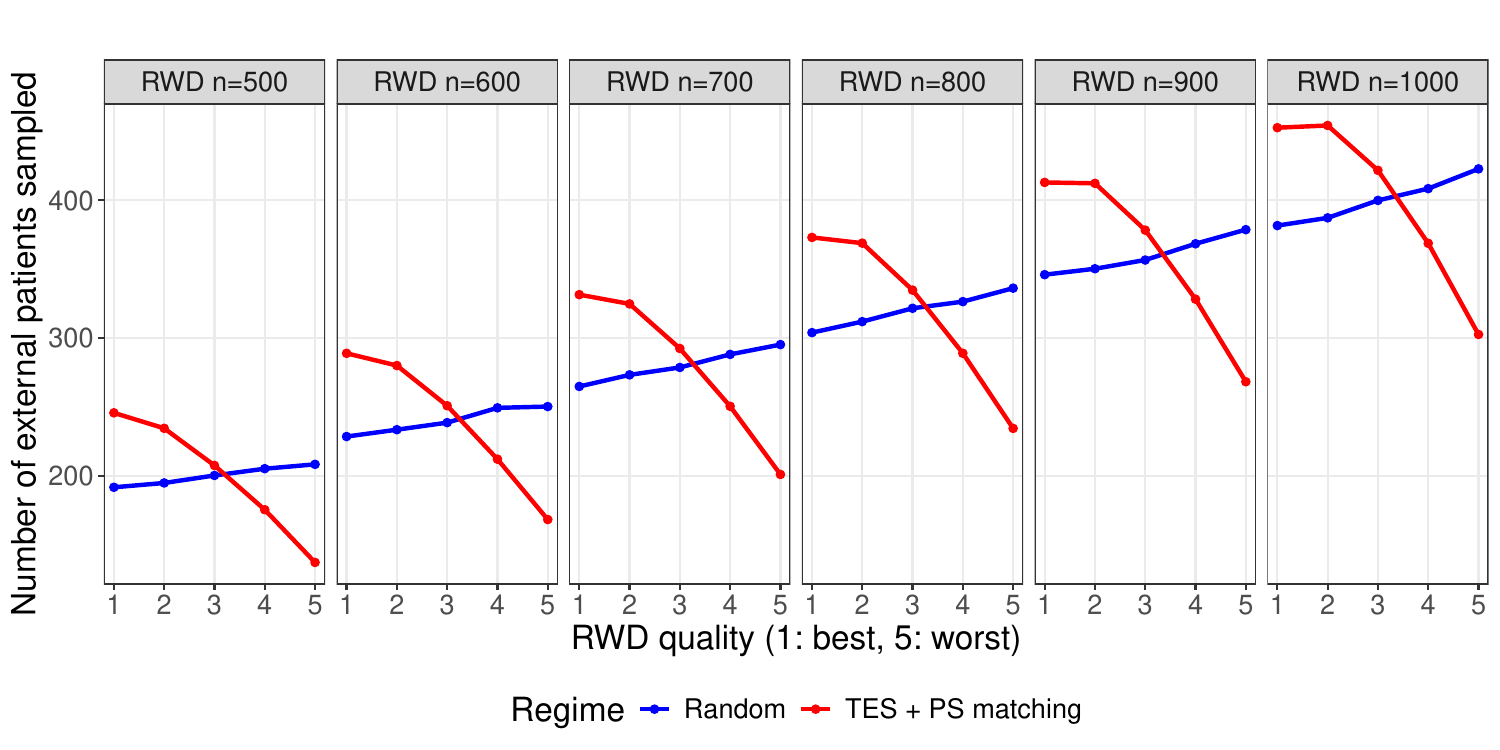}
\caption{The distribution of the number of external patients sampled across varying quality (i.e. different levels of bias) of external RWD sources under each sampling regime. TES + PS matching technique resulted in selecting more patients with less bias and fewer patients with large bias.}
\label{fig:random_vs_ideal_rwd_quality}
\end{figure}

We compared two strategies for sampling an external cohort of a desired sample size. The first is a random sampling approach, where external patients are randomly drawn from the available RWD pool. The second was our proposed strategy, which involves trial enrollment score matching followed by propensity score matching to select external patients. Specifically, we used $k=30$ and $m=1$: each RCT subject was matched with 30 external subjects from the larger pool of 20,000 subjects based on the estimated trial enrollment probability using nearest-neighbor matching \citep{ho_matchit_2011}. Next, within the subset of selected external patients, we estimated the propensity score and matched each treated external subject with one external control based on the estimated propensity score. To ensure a fair comparison, the sample size of the selected external cohort was kept the same between the random sampling and matching strategies. We then applied A-TMLE estimators to the pooled RCT data and the selected external RWD under each strategy. As a benchmark, we also evaluated estimates derived solely from the RCT data using a standard TMLE. Note that one may also apply an unadjusted estimator on the RCT as the benchmark estimator, but this would be less efficient than TMLE \citep{moore_increasing_2009}.

\begin{table}[h!]
\centering
\resizebox{\columnwidth}{!}{
\begin{tabular}{ccccccc}
\toprule
Sample size of external arms & Strategy & $|\text{Bias}|$ & Variance & 95\% CI widths & 95\% CI coverage & Power \\
\midrule
\multirow{1}{*}{0} & RCT data only & 0.002 & 0.081 & 1.169 & 0.97 & 0.36 \\
\midrule
\multirow{2}{*}{500} 
& Random & 0.119 & 0.103 & 0.959 & 0.86 & 0.40 \\
& TES + PS matching & \textbf{0.007} & \textbf{0.041} & \textbf{0.725} & \textbf{0.95} & \textbf{0.85} \\
\midrule
\multirow{2}{*}{600} 
& Random & 0.170 & 0.084 & 0.923 & 0.85 & 0.35 \\
& TES + PS matching & \textbf{0.028} & \textbf{0.045} & \textbf{0.698} & \textbf{0.94} & \textbf{0.87} \\
\midrule
\multirow{2}{*}{700} 
& Random & 0.156 & 0.093 & 0.905 & 0.79 & 0.39 \\
& TES + PS matching & \textbf{0.016} & \textbf{0.045} & \textbf{0.680} & \textbf{0.94} & \textbf{0.89} \\
\midrule
\multirow{2}{*}{800} 
& Random & 0.157 & 0.103 & 0.880 & 0.78 & 0.39 \\
& TES + PS matching & \textbf{0.012} & \textbf{0.046} & \textbf{0.679} & \textbf{0.95} & \textbf{0.89} \\
\midrule
\multirow{2}{*}{900} 
& Random & 0.134 & 0.080 & 0.881 & 0.84 & 0.42 \\
& TES + PS matching & \textbf{0.011} & \textbf{0.040} & \textbf{0.618} & \textbf{0.96} & \textbf{0.90} \\
\midrule
\multirow{2}{*}{1000} 
& Random & 0.167 & 0.085 & 0.859 & 0.80 & 0.38 \\
& TES + PS matching & \textbf{0.019} & \textbf{0.040} & \textbf{0.605} & \textbf{0.95} & \textbf{0.93} \\
\bottomrule
\end{tabular}
}
\caption{Absolute bias, variance, average 95\% confidence interval widths, coverage, and power by external data sample size and sampling strategy. ``RCT data only" denotes the RCT-only design, analyzed with the TMLE estimator. ``Random" denotes a randomly sampled external cohort of the specified size, analyzed with the A-TMLE estimator. ``TES + PS matching" denotes our proposed trial enrollment score and propensity score matching strategy for sampling the external cohort of the specified size, also analyzed with the A-TMLE estimator.}
\label{tab:atmle}
\end{table}

Figure \ref{fig:random_vs_ideal_rwd_quality} compares the number of patients selected from each external data source under the two sampling strategies. As expected, our proposed strategy predominantly selects patients from the least biased external data sources, while minimizing the selection of patients from the most biased sources. After selecting the external cohort, A-TMLE is applied to perform the data integration analysis. Table \ref{tab:atmle} presents the results, including absolute bias, variance, 95\% confidence interval widths, coverage, and statistical power of the three design strategies under comparison, evaluated across increasing sample sizes of the selected external cohort. We observe that under our proposed trial enrollment and propensity score matching strategy, the A-TMLE estimator had narrower confidence intervals, nominal coverage, and better statistical power to detect the difference. For the random sampling strategy, we observe that A-TMLE has below-nominal coverage. This is expected, as we intentionally used a main-term lasso rather than full HAL bases for learning the working models in A-TMLE, which results in model misspecification of the conditional average RCT-enrollment effect $\tau_{S,0}(W,A)$. However, the robustness structure of A-TMLE should recover some robustness if the trial enrollment score and the propensity score in the external data are close to constants as a result of our proposed matching-based sampling strategy. Indeed, we observe that not only is the bias reduced (as inferred from the simultaneously narrower confidence intervals and improved coverage), but the confidence interval coverage is also back to nominal. In Appendix \ref{app:additional_sims}, we present additional simulation results using ES-CVTMLE \citep{dang_escvtmle_2022}, another data integration estimator, under the same simulation setup. These findings align with those observed for A-TMLE. Specifically, the ES-CVTMLE estimator had smaller bias and variance, narrower confidence intervals, and better coverage under the proposed matching-based design compared to the random sampling regime.

\section{Case Study: Augmenting the DEVOTE Trial}\label{sec:case_study}
We applied our proposed matching strategy to select external real-world patients from the Optum Clinformatics claims database to augment the DEVOTE trial. The DEVOTE trial was a double-blind, treat-to-target, event-driven study designed to assess the cardiovascular safety of degludec, an ultralong-acting, once-daily basal insulin approved by the FDA for use in various diabetic populations \citep{marso_efficacy_2017}. The primary objective of the trial was to evaluate whether degludec posed no greater cardiovascular risk than glargine, with a noninferiority hazard ratio margin of 1.3 (allowing up to 30\% excess risk). The treatment arm ($n=3,818$) comprised patients receiving insulin degludec, while the comparator arm ($n=3,817$) consisted of patients treated with insulin glargine. The primary endpoint was a composite measure of major adverse cardiac events (MACE), including cardiovascular death, nonfatal myocardial infarction, and nonfatal stroke. In the methodology paper on the A-TMLE for data integration, we used the same trial data and external data source to demonstrate an application of the A-TMLE method \citep{adaptive_tmle_2024}. The difference is that in that analysis, all external patients meeting the trial's inclusion and exclusion criteria were included without matching. Here, we extend that case study by incorporating our proposed external patient selection method to further improve the robustness of the data integration analysis. For the current case study, we focus on the 540- and 730-day (1.5- and 2-year) risk differences of MACE between the two study arms. In \citep{adaptive_tmle_2024}, we also highlighted the significance of augmenting this particular trial with external RWD. Specifically, as a noninferiority trial, DEVOTE was not originally designed to test for superiority. Data integration therefore provides a way to explore the opportunity for detecting treatment superiority by leveraging external RWD, avoiding the need for a separate, dedicated superiority trial, which may be time-consuming and costly especially for studying relatively rare endpoints like MACE.

\subsection{Selection of external real-world patients}
To address time-related bias, we restricted the sampling of external data to the period between 2013 and 2019. Although the trial concluded in 2017, we extended the sampling window by two additional years to increase the inclusion of insulin degludec (treatment arm) patients, as relatively few patients were using this drug by 2017 due to its recent approval. After restricting the time frame, we applied the same inclusion and exclusion criteria as the trial to filter external patients, aiming to emulate the trial's patient enrollment process as closely as possible. 

Next, we implemented our proposed matching strategy to refine the external cohort for the final data integration analysis. We began by identifying covariates that could predict treatment assignment in the external data, including age, sex, low-density lipoprotein, serum creatinine levels, hemoglobin A1C, and others. We then performed trial enrollment score matching as the first step. Specifically, we fitted a logistic regression on the pooled RCT and external data, regressing the study indicator $S$ on these covariates to estimate the trial enrollment score for all patients. Using these estimated scores, we conducted 1:3 nearest-neighbor matching and selected a total of $22,905$ candidate patients from the Optum database, which is exactly three times the sample size of the RCT. Following this, we performed the second stage of matching, propensity score matching within the external data. We fitted another logistic regression, this time regressing treatment status $A$ on the same set of covariates, but only using the selected subset of external data from the previous trial enrollment score matching. Based on the estimated propensity scores, we conducted 1:5 nearest-neighbor matching. The rationale for including more external control patients is the greater availability of control patients in the external data, as the comparator drug had been on the market for a longer period than insulin degludec. We aimed to maximize the use of these control patients without unnecessary exclusion. The resulting cohort, comprising $2,359$ treated and $11,795$ control patients, constitutes the final dataset for data integration analysis.
\begin{figure}[h!]
\centering
\includegraphics[width=0.8\linewidth]{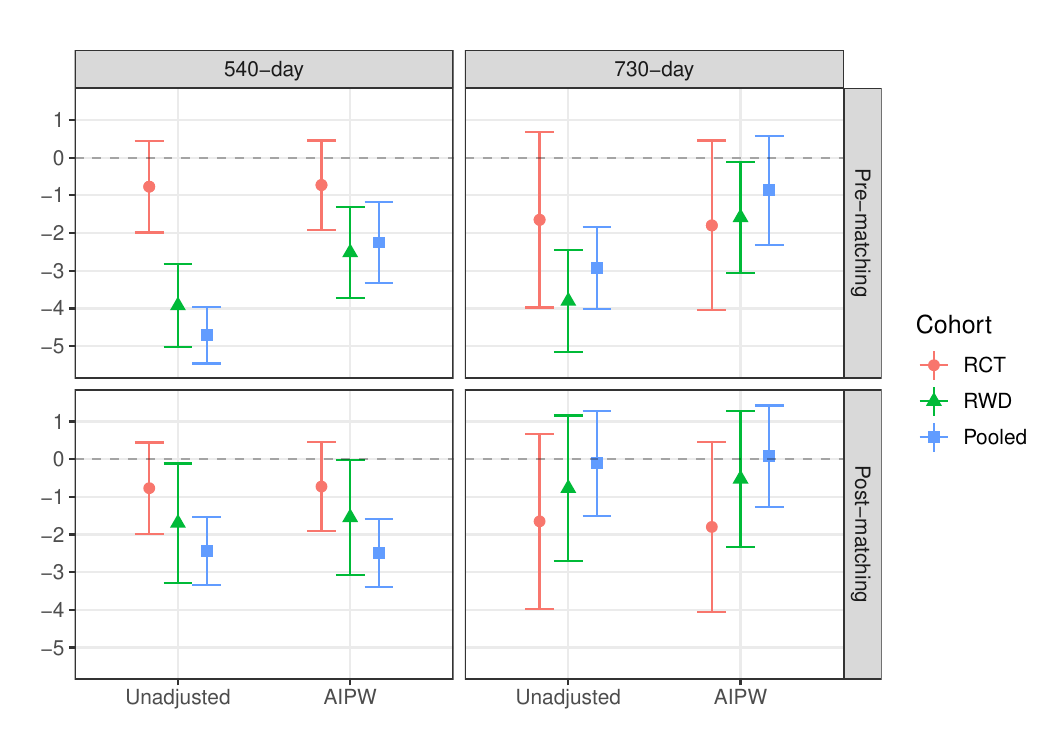}
\caption{Unadjusted and adjusted point estimates with corresponding 95\% confidence intervals for the 540-day and 730-day risk difference (\%) of MACE, based on the RCT, RWD, and pooled cohorts, both pre- and post-matching. The adjusted estimates are obtained using the augmented inverse probability weighting (AIPW) estimator.}
\label{fig:preliminary}
\end{figure}

\subsection{Preliminary analysis}
We began with a preliminary analysis without applying any data integration estimators. Specifically, we compared the unadjusted and adjusted (AIPW) estimates on the RCT, RWD versus the pooled data obtained before and after implementing our matching strategy. Note that the RWD pre-matching includes all candidate patients after the initial time frame restriction and inclusion and exclusion criteria filtering, whereas the RWD patients post-matching includes only those who were matched by our design. The results are presented in Figure \ref{fig:preliminary}.

The pre-matching results show a noticeable discrepancy between the estimates derived from the RWD (or pooled data) and those from the RCT. Since the RCT is expected to provide an unbiased estimate of the causal effect, this discrepancy can be attributed to two factors. First, the target parameters for the RCT, RWD, and pooled data differ. Specifically, the RWD estimates reflect the real-world effect, the RCT estimates capture the trial effect, and the pooled data estimates reflect the effect in the combined population. Differences in the populations under study can naturally lead to differences in the estimated effects. Second, the discrepancy may stem from potential violations of identification assumptions in the RWD. For instance, unmeasured confounding in the RWD could bias the estimates and explain their divergence from the RCT results. This difference is particularly pronounced at 540 days, where the unadjusted RWD estimates do not overlap with the RCT estimates. Although adjusting for potential confounders narrows the gap, the difference remains substantial. In contrast, after performing the matching and selecting a subset of external patients, the point estimates are closer to the RCT estimates, suggesting that the matching strategy helps reduce discrepancies to some extent. 
\begin{figure}[h!]
\centering
\includegraphics[width=0.8\linewidth]{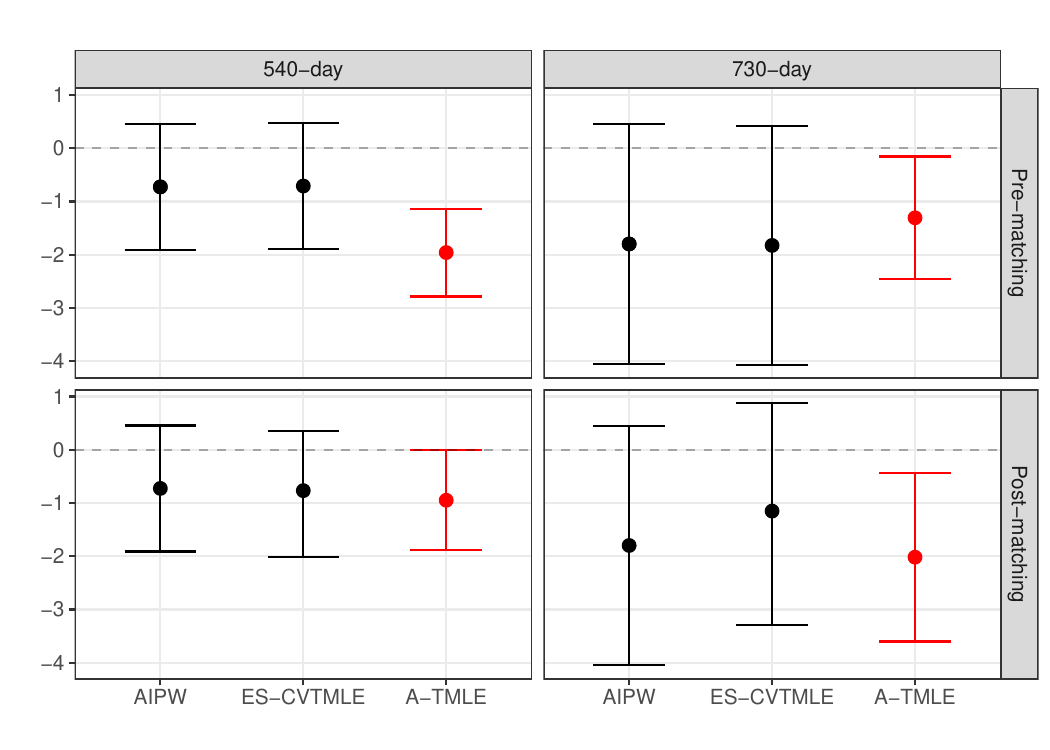}
\caption{Data integration analysis results. ES-CVTMLE and A-TMLE point estimates with corresponding 95\% confidence intervals for the 540-day and 730-day risk difference (\%) of MACE, both pre- and post-matching. AIPW estimates are based on RCT data only.}
\label{fig:data_integration_analysis}
\end{figure}
\begin{table}[h!]
\centering
\resizebox{\linewidth}{!}{
\begin{tabular}{ccc|cc}
\toprule
& & & \multicolumn{2}{c}{Estimated risk difference (\%) of MACE with 95\% CIs} \\
Pre/Post-matching & Estimator & Study design & 540-day & 730-day \\
\midrule
- & AIPW & DEVOTE & -0.726 (-1.910, 0.459) & -1.798 (-4.050, 0.455) \\
\midrule
\multirow{2}{*}{Pre} & ES-CVTMLE & DEVOTE + Optum & -0.709 (-1.893, 0.475) & -1.825 (-4.071, 0.420) \\ 
& A-TMLE & DEVOTE + Optum & \textbf{-1.959 (-2.783, -1.134)} & \textbf{-1.306 (-2.459, -0.153)} \\ 
\midrule
\multirow{2}{*}{Post} & ES-CVTMLE & DEVOTE + Optum & -0.767 (-2.008, 0.360) & -1.152 (-3.298, 0.885) \\ 
& A-TMLE & DEVOTE + Optum & \textbf{-0.948 (-1.889, -0.008)} & \textbf{-2.018 (-3.601, -0.436)} \\ 
\bottomrule
\end{tabular}
}
\caption{Estimated risk differences (\%) of MACE between the insulin degludec and insulin glargine groups, with corresponding 95\% confidence intervals. Statistically significant estimates under $\alpha=0.05$ for a two-sided test are marked in bold.}
\label{tab:devote_optum}
\end{table}

\subsection{Data integration analysis}
However, as suggested by Figure \ref{fig:preliminary}, matching alone cannot address all potential biases. To ensure robust and unbiased estimates, it is crucial to apply estimators that do not rely on assumptions beyond those that are known to be true by the design of the RCT. We therefore applied data integration estimators ES-CVTMLE and A-TMLE to the pooled data to handle any residual biases. The results after applying data integration estimators are presented in Figure \ref{fig:data_integration_analysis}, where we compare estimates obtained before and after applying the matching strategy. The values for the point estimates and corresponding 95\% confidence intervals are shown in Table \ref{tab:devote_optum}. The results show that the confidence interval widths of A-TMLE are significantly smaller compared to both the RCT-only and ES-CVTMLE estimates, and A-TMLE achieves statistical significance for testing superiority. Comparing the pre- and post-matching estimates, we observe that at the 540-day mark, using all external data without matching results in an A-TMLE point estimate that lies outside the RCT estimate's confidence interval. After implementing the matching step, the A-TMLE point estimate aligns much more closely with the RCT estimate, suggesting that matching effectively helps reduce bias in this case. These findings show the importance of combining matching with a robust data integration estimator. While matching alone can reduce potential bias from covariates, applying an estimator like A-TMLE after matching serves as an additional safeguard, ensuring that the results are more reliable. We therefore recommend pairing matching with a data integration estimator to enhance the validity of the analysis.

\section{Discussion}\label{sec:discussion}
In this article, we analyzed the robustness structure of the A-TMLE estimator and proposed a matching-based design to select external patients in a way that leverages these robustness features. Specifically, our analysis of the exact remainder revealed that if the trial enrollment score, $P_0(S=1\mid W)$, and the propensity score in the external data, $P_0(A=1\mid S=0,W)$, are fully known, the exact remainder of the projection parameter becomes zero. Furthermore, if these scores can be arranged by design to be constant in $W$, the oracle bias would also be zero, meaning that we are robust even with respect to the nonparametrically defined target parameter. Motivated by this analysis, we proposed a two-step matching design, where we first match based on the estimated trial enrollment score to select a subset of external patients, and then, within this subset, we perform propensity score matching. Empirical simulations demonstrated that this approach, in scenarios where the bias working model is misspecified, reduces the MSE of point estimates, improves confidence interval coverage, and increases statistical power of the A-TMLE estimator. Applying this methodology to augment the DEVOTE trial with external patients from the Optum database showed that this matching strategy could align the estimates obtained from data integration estimators more closely with those from the RCT, suggesting an improvement in the robustness of data integration estimators like A-TMLE. Additionally, the case study demonstrated the importance of applying robust estimators after matching, as residual biases may persist that cannot be fully addressed by matching alone. 

We would like to point out a few limitations of our approach. First, one may be concerned that the matched sample introduces a selection bias, as individuals are no longer independently and identically drawn from the target population. However, we also note that RCT samples are often not drawn i.i.d from the target population as well. If the priority of data integration is to improve the power of the treatment effect estimators in the trials, then this biased sampling might act in our favor by adding more robustness to the estimation of treatment effects within the trial augmented with external RWD. Second, in order to fully exploit the robustness structure of A-TMLE, one needs to ensure that the covariates are exactly matched, which may be difficult to achieve in practice if covariates are continuous or high-dimensional. Nonetheless, the value of this matching approach lies in its practical potential to improve control over the trial enrollment and propensity scores in the sampled subset of external data. As discussed in Section \ref{sec:design}, this would potentially allow for a slightly slower convergence rate for the working models without compromising valid inference, thereby making the estimator more robust. In addition, the matching steps may simplify the bias working model by reducing bias due to covariates, which is often a major source of overall bias. If the covariates are effectively balanced after the matching, the true bias is likely to be smaller and potentially more parsimonious. As highlighted in \citep{adaptive_tmle_2024}, A-TMLE benefits from scenarios where the bias is small. In such cases, even a simpler working model may approximate the bias well enough, despite the true bias being fully nonparametric. This can lead to additional efficiency gains for A-TMLE, especially in finite samples. While it may seem intuitive that larger external datasets always improve estimation, this is not necessarily the case in finite samples. Well-matched subsets of external data can often yield better finite sample performance by reducing the burden of the bias correction.

Our proposed design is completely outcome-blind. When estimating the trial enrollment and propensity scores, no outcome data is used, avoiding the risk of cherry-picking external patients with favorable outcomes. The procedure can also be fully pre-specified. Additionally, the approach is highly flexible. Researchers can experiment with different matching algorithms, each producing a subset of external patients, and rank these subsets based on how well they balance covariates. The estimated trial enrollment score can serve as a criterion for evaluating the quality of the matches. In conclusion, our proposed strategy offers a practical framework for augmenting RCTs with external RWD while mitigating potential biases. By applying a matching-based design motivated to make the estimator more robust, paired with robust estimation methods like A-TMLE, this approach has the potential to improve both the accuracy and efficiency of causal inferences in data integration settings.

\newpage
\bibliography{references}

\appendix
\section{Proof of Lemmas}\label{app_R2}
\subsection{Proof of Lemma \ref{lem:param_unbiased}}
\begin{proof}
We define the exact remainders
\begin{align*}
R_{\mathcal{M}_{w}}(P,P_0) &\equiv \Psi_{\mathcal{M}_{w}}(P) - \Psi_{\mathcal{M}_{w}}(P_0) + P_0 D_{\mathcal{M}_{w},P}, \\
\tilde{R}_{\mathcal{M}_{A,w}}(P,P_0) &\equiv \tilde{\Psi}_{\mathcal{M}_{A,w}}(P) - \tilde{\Psi}_{\mathcal{M}_{A,w}}(P_0) + P_0 D_{\mathcal{M}_{A,w},P}, \\
R^\#_{\mathcal{M}_{S,w}}(P,P_0) &\equiv \Psi^\#_{\mathcal{M}_{S,w}}(P) - \Psi^\#_{\mathcal{M}_{S,w}}(P_0) + P_0 D_{\mathcal{M}_{S,w},P}.
\end{align*}
As discussed in Section \ref{sec:design}, the robustness of an estimator is closely tied to the behavior of the exact remainder. To gain insights into the robustness of the A-TMLE estimators, we derive the remainders $\tilde{R}_{\mathcal{M}_{A,w}}(P,P_0)$ and $R^\#_{\mathcal{M}_{S,w}}(P,P_0)$, which correspond to the remainders for the pooled-ATE projection parameter and the bias projection parameter, respectively.

\subsubsection*{Exact remainder of the pooled-ATE projection parameter $\tilde{\Psi}_{\mathcal{M}_{A,w}}$}
Let $\beta$ represent a vector of real-valued coefficients, and let $\phi(W)$ denote the design matrix consisting of $p$ basis functions. Let $\phi_{j}(W)$ and $\beta(j)$ denote the $j$-th basis function and its corresponding coefficient, respectively. Recall that
$$
\tilde{\Psi}_{\mathcal{M}_{A,w}}(P)=E_P[\tau_{A,\beta_P}(W)]=E_P\sum_j\beta_P(j)\phi_{j}(W).
$$
We derived in lemmas 4 and 5 in prior work \citep{adaptive_tmle_2024} that the canonical gradient of $\tilde{\Psi}_{\mathcal{M}_{A,w}}$ at $P$ is given by
$$
D_{\tilde{\Psi},\mathcal{M}_{A,w},P}=\tau_{A,\beta_P}(W)-\tilde{\Psi}_{\mathcal{M}_{A,w}}(P)+\sum_jD_{\beta,P,j}E_P\phi_{j}(W),
$$
where
$$
D_{\beta,P}=(A-g_P(1\mid W))\tilde{I}_P^{-1}\phi(W)(Y-\theta_P(W)-(A-g_P(1\mid W))\tau_{A,\beta_P}(W)),
$$
and
$$
\tilde{I}_{P}=E_Pg(1-g)(1\mid W)\phi\phi^\top(W).
$$
We now compute the exact remainder
$$
\tilde{R}_{\mathcal{M}_{A,w}}(P,P_0)=\tilde{\Psi}_{\mathcal{M}_{A,w}}(P)-\tilde{\Psi}_{\mathcal{M}_{A,w}}(P_0)+P_0D_{\tilde{\Psi},\mathcal{M}_{A,w},P}.
$$
Note that,
\begin{align*}
\tilde{\Psi}_{\mathcal{M}_{A,w}}(P)+P_0D_{\tilde{\Psi},\mathcal{M}_{A,w},P}
&=P_0\tau_{A,\beta_P}(W)+P_0\left[\sum_jD_{\beta,P,j}E_P\phi_{j}(W)\right]\\
&=P_0\tau_{A,\beta_P}(W)+\sum_jE_P\phi_{j}(W)P_0D_{\beta,P,j}.
\end{align*}
In particular, since $\bar{Q}_0(W,A)=\theta_0(W)+(A-g_0(1\mid W))\tau_{A,\beta_0}(W)$,
\begin{align*}
P_0D_{\beta,P}=\tilde{I}_{P}^{-1}P_0\bigg[(A-g_P(1\mid W))\phi(W)(&\bar{Q}_0(W,A)-\theta_P(W)-(A-g_P(1\mid W))\tau_{A,\beta_P}(W))\bigg]\\
=\tilde{I}_{P}^{-1}P_0\bigg[(A-g_P(1\mid W))\phi(W)(&\theta_0(W)+(A-g_0(1\mid W))\tau_{A,\beta_0}(W)\\
-&\theta_P(W)-(A-g_P(1\mid W))\tau_{A,\beta_P}(W))\bigg].
\end{align*}
Let's then define
$$
P_0D_{\beta,P}=\tilde{I}_{P}^{-1}(\mathbf{A}+\mathbf{B}+\mathbf{C}),
$$
where
\begin{align*}
\mathbf{A}&\equiv P_0\bigg[(A-g_P(1\mid W))\phi(W)(\theta_0-\theta_P)(W)\bigg]\\
&=P_0\bigg[(g_0-g_P)(1\mid W)\phi(W)(\theta_0-\theta_P)(W)\bigg];
\end{align*}
\begin{align*}
\mathbf{B}&=P_0\bigg[(A-g_P(1\mid W))\phi(W)A(\tau_{A,\beta_0}-\tau_{A,\beta_P})(W)\bigg]\\
&=P_0\bigg[g_0(1-g_P)(1\mid W)\phi(W)(\tau_{A,\beta_0}-\tau_{A,\beta_P})(W)\bigg]\\
&=P_0\bigg[(g_0-g_P)(1-g_P)(1\mid W)\phi(W)(\tau_{A,\beta_0}-\tau_{A,\beta_P})(W)\bigg]\\
&+P_0\bigg[g_P(1-g_P)(1\mid W)\phi(W)(\tau_{A,\beta_0}-\tau_{A,\beta_P})(W)\bigg];
\end{align*}
and
\begin{align*}
\mathbf{C}&=P_0\bigg[(A-g_P(1\mid W))\phi(W)(g_P(1\mid W)\tau_{A,\beta_P}(W)-g_0(1\mid W)\tau_{A,\beta_0}(W))\bigg]\\
&=P_0\bigg[(A-g_P(1\mid W))\phi(W)(g_P-g_0)(1\mid W)\tau_{A,\beta_P}(W)\bigg]\\
&+P_0\bigg[(A-g_P(1\mid W))\phi(W)g_0(1\mid W)(\tau_{A,\beta_P}-\tau_{A,\beta_0})(W)\bigg]\\
&=-P_0\bigg[(g_P-g_0)^2(1\mid W)\phi(W)\tau_{A,\beta_P}(W)\bigg]\\
&+P_0\bigg[(g_0-g_P)(1\mid W)\phi(W)g_0(1\mid W)(\tau_{A,\beta_P}-\tau_{A,\beta_0})(W)\bigg].
\end{align*}
Since $\mathbf{A}$ and $\mathbf{C}$ have second-order differences in $P-P_0$, we focus on $\mathbf{B}$. The first term of $\mathbf{B}$ is second-order. Let $\mathbf{B}_{2}$ denote the second term of $\mathbf{B}$, i.e.,
$$
\mathbf{B}_{2}\equiv P_0\bigg[g_P(1-g_P)(1\mid W)\phi(W)(\tau_{A,\beta_0}-\tau_{A,\beta_P})(W)\bigg].
$$
Since $\tilde{I}_{P}=E_Pg_P(1-g_P)(1\mid W)\phi\phi^\top(W)$, we have
$$
\sum_jE_P\phi_{j}(W)(\tilde{I}_{P}^{-1}\mathbf{B}_{2})_j=P_0\tau_{A,\beta_0}(W)-P_0\tau_{A,\beta_P}(W)=\tilde{\Psi}_{\mathcal{M}_{A,w}}(P_0)-P_0\tau_{A,\beta_P}(W).
$$
To summarize, the exact remainder $\tilde{R}_{\mathcal{M}_{A,w}}(P,P_0)$ is
\begin{align*}
\tilde{R}_{\mathcal{M}_{A,w}}(P,P_0)=\sum_jE_P\phi_j(W)\{\tilde{I}_{P}^{-1}
P_0[&(g_P-g_0)(1\mid W)\phi(W)(\theta_P-\theta_0)(W)\\
+&(g_P-g_0)(1-g_P)(1\mid W)\phi(W)(\tau_{A,\beta_P}-\tau_{A,\beta_0})(W)\\
-&(g_P-g_0)^2(1\mid W)\phi(W)\tau_{A,\beta_P}(W)\\
+&(g_P-g_0)(1\mid W)\phi(W)g_0(1\mid W)(\tau_{A,\beta_P}-\tau_{A,\beta_0})(W)]\}_j.
\end{align*}
In particular, note that if $g_P=g_0$, then the exact remainder $\tilde{R}_{\mathcal{M}_{A,w}}(P,P_0)=0$, even if $\theta$ and $\tau_{A}$ are misspecified.

\subsubsection*{Exact remainder of the bias projection parameter $\Psi^\#_{\mathcal{M}_{S,w}}$}
Now, let $\beta$ represent a vector of real-valued coefficients, and let $\phi(W,A)$ denote the design matrix consisting of $p$ basis functions. We reuse the notations $\beta$ and $\phi$ here, although they differ from those used previously for $\tilde{\Psi}_{\mathcal{M}_{A,w}}$. The context, i.e., whether referring to $\tilde{\Psi}_{\mathcal{M}_{A,w}}$ or $\Psi^\#_{\mathcal{M}_{S,w}}$, will distinguish their meanings. Let $\phi_{j}(W,A)$ and $\beta(j)$ denote the $j$-th basis function and its corresponding coefficient, respectively. Recall that
$$
\Psi^\#_{\mathcal{M}_{S,w}}(P)=E_P\Pi_P(0\mid W,0)\tau_{S,\beta_{P}}(W,0)-E_P\Pi_P(0\mid W,1)\tau_{S,\beta_{P}}(W,1).
$$
We derived in lemmas 7 and 8 in prior work \citep{adaptive_tmle_2024} that the canonical gradient of $\Psi^\#_{\mathcal{M}_{S,w}}$ at $P$ is given by
$$
D^{\#}_{{\cal M}_{S,w},P}=D^{\#}_{{\cal M}_{S,w},P_W,P}+D^{\#}_{{\cal M}_{S,w},\Pi,P}+D^{\#}_{{\cal M}_{S,w},\beta,P},
$$
where
\begin{align*}
D^{\#}_{{\cal M}_{S,w},P_W,P}&= \Pi_P(0\mid W,0)\tau_{S,\beta_P}(W,0)-\Pi_P(0\mid W,1)\tau_{S,\beta_P}(W,1)-\Psi^{\#}_{{\cal M}_{S,w}}(P)\\
D^{\#}_{{\cal M}_{S,w},\Pi,P}&=\left[\frac{A}{g_P(1\mid W)}\tau_{S,\beta_P}(W,1)-\frac{1-A}{g_P(0\mid W)}\tau_{S,\beta_P}(W,0)\right](S-\Pi_P(1\mid W,A))\\
D^{\#}_{{\cal M}_{S,w},\beta,P}&=\sum_j D_{\beta,P,j} E_P\Pi_P(0\mid W,0)\phi_{j}(W,0)-\sum_j D_{\beta,P,j} E_P \Pi_P(0\mid W,1)\phi_{j}(W,1),
\end{align*}
where
$$
D_{\beta,P}=I_{P}^{-1}(S-\Pi_P(1\mid W,A))\phi(W,A)(Y-\bar{Q}_P(W,A)-(S-\Pi_P(1\mid W,A))\tau_{S,\beta_P}(W,A)),
$$
and
$$
I_P=E_P \Pi(1-\Pi)(1\mid W,A){\bf \phi}{\bf \phi}^{\top}(W,A)
$$
First note that, for the $P_W$-component, we have 
$$
P_0D^{\#}_{{\cal M}_{S,w},P_W,P}=P_0\left[\Pi_P(0\mid W,0)\tau_{S,\beta_P}(W,0)\right]-P_0\left[\Pi_P(0\mid W,1)\tau_{S,\beta_P}(W,1)\right]-P_0\Psi^{\#}_{{\cal M}_{S,w}}(P).
$$
For the $\Pi$-component, we have
\begin{align*}
&P_0D^{\#}_{{\cal M}_{S,w},\Pi,P}\\
&=P_0\left[\left(\frac{A}{g_P(1\mid W)}\tau_{S,\beta_{P}}(W,1)-\frac{1-A}{g_P(0\mid W)}\tau_{S,\beta_{P}}(W,0)\right)(\Pi_0-\Pi_P)(1\mid W,A)\right]\\
&=P_0\left[\frac{g_0}{g_P}(1\mid W)\tau_{S,\beta_{P}}(W,1)(\Pi_0-\Pi_P)(1\mid W,1)-\frac{g_0}{g_P}(0\mid W)\tau_{S,\beta_{P}}(W,0)(\Pi_0-\Pi_P)(1\mid W,0)\right]\\
&=P_0\left[\frac{g_0-g_P}{g_P}(1\mid W)\tau_{S,\beta_{P}}(W,1)(\Pi_0-\Pi_P)(1\mid W,1)\right]+P_0\left[\tau_{S,\beta_{P}}(W,1)(\Pi_0-\Pi_P)(1\mid W,1)\right]\\
&-P_0\left[\frac{g_0-g_P}{g_P}(0\mid W)\tau_{S,\beta_{P}}(W,0)(\Pi_0-\Pi_P)(1\mid W,0)\right]-P_0\left[\tau_{S,\beta_{P}}(W,0)(\Pi_0-\Pi_P)(1\mid W,0)\right].
\end{align*}
Note that,
\begin{align*}
&-P_0\left[\tau_{S,\beta_{P}}(W,0)(\Pi_0-\Pi_P)(1\mid W,0)\right]\\
&=-P_0\left[\tau_{S,\beta_{P}}(W,0)(\Pi_P-\Pi_0)(0\mid W,0)\right]\\
&=-P_0\left[(\tau_{S,\beta_{P}}-\tau_{S,\beta_{0}})(W,0)(\Pi_P-\Pi_0)(0\mid W,0)\right]-P_0\left[\tau_{S,\beta_0}(W,0)(\Pi_P-\Pi_0)(0\mid W,0)\right].
\end{align*}
Similarly,
\begin{align*}
&P_0\left[\tau_{S,\beta_{P}}(W,1)(\Pi_0-\Pi_P)(1\mid W,1)\right]\\
&=P_0\left[\tau_{S,\beta_{P}}(W,1)(\Pi_P-\Pi_0)(0\mid W,1)\right]\\
&=P_0\left[(\tau_{S,\beta_{P}}-\tau_{S,\beta_0})(W,1)(\Pi_P-\Pi_0)(0\mid W,1)\right]+P_0\left[\tau_{S,\beta_0}(W,1)(\Pi_P-\Pi_0)(0\mid W,1)\right].
\end{align*}
For the $\beta$-component, we define
$$
D^{\#}_{\mathcal{M}_{S,w},\beta,P}\equiv D^{\#}_{\mathcal{M}_{S,w},\beta,P,A_0}-D^{\#}_{\mathcal{M}_{S,w},\beta,P,A_1},
$$
where
$$
D^{\#}_{\mathcal{M}_{S,w},\beta,P,A_0}\equiv \sum_j D_{\beta,P,j} E_P\Pi_P(0\mid W,0)\phi_{j}(W,0)
$$
and
$$
D^{\#}_{\mathcal{M}_{S,w},\beta,P,A_1}\equiv \sum_j D_{\beta,P,j} E_P \Pi_P(0\mid W,1)\phi_{j}(W,1).
$$
We first focus on the control-arm term $D^{\#}_{\mathcal{M}_{S,w},\beta,P,A_0}$. Note that we can reparameterize $Q_0(S,W,A)$ as $\bar{Q}_0(W,A)+(S-\Pi_0(1\mid W,A))\tau_{S,\beta_0}(W,A)$. Therefore,
\begin{align*}
P_0D_{\beta,P}=I_{P}^{-1}P_0\big[(S-\Pi_P(1\mid W,A))\phi(W,A)(&E_0(Y\mid S,W,A)-\bar{Q}_P(W,A)-(S-\Pi_P(1\mid W,A))\tau_{S,\beta_{P}}(W,A))\big]\\
=I_{P}^{-1}P_0\big[(S-\Pi_P(1\mid W,A))\phi(W,A)(&\bar{Q}_0(W,A)+(S-\Pi_0(1\mid W,A))\tau_{S,\beta_0}(W,A)\\
-&\bar{Q}_P(W,A)-(S-\Pi_P(1\mid W,A))\tau_{S,\beta_{P}}(W,A))\big].
\end{align*}
We define
\begin{align*}
D^{\#}_{\mathcal{M}_{S,w},\beta,P,A_0}&=\sum_jE_P\Pi_P(0\mid W,0)\phi_j(W,0)\left[I_{P}^{-1}(\mathbf{A}+\mathbf{B}+\mathbf{C})\right]_j,
\end{align*}
where
\begin{align*}
\mathbf{A}&\equiv P_0\left[(S-\Pi_P(1\mid W,A))\phi(W,A)(\bar{Q}_0-\bar{Q}_P)(W,A)\right]\\
&=P_0\left[(\Pi_P-\Pi_0)(1\mid W,A)\phi(W,A)(\bar{Q}_P-\bar{Q}_0)(W,A)\right];\\
\mathbf{B}&\equiv P_0\left[(S-\Pi_P(1\mid W,A))\phi(W,A)S(\tau_{S,\beta_0}-\tau_{S,\beta_P})(W,A)\right]\\
&=P_0\left[\Pi_0(1-\Pi_P)(1\mid W,A)\phi(W,A)
(\tau_{S,\beta_0}-\tau_{S,\beta_P})(W,A)\right]\\
&=P_0\left[(\Pi_0-\Pi_P)(1-\Pi_P)(1\mid W,A)\phi(W,A)(\tau_{S,\beta_0}-\tau_{S,\beta_P})(W,A)\right]\\
&+P_0\left[\Pi_P(1-\Pi_P)(1\mid W,A)\phi(W,A)(\tau_{S,\beta_0}-\tau_{S,\beta_P})(W,A)\right];\\
\mathbf{C}&\equiv P_0\bigg[(S-\Pi_P(1\mid W,A))\phi(W,A)(\Pi_P(1\mid W,A)\tau_{S,\beta_P}(W,A)-\Pi_0(1\mid W,A)\tau_{S,\beta_0}(W,A))\bigg]\\
&=P_0\bigg[(\Pi_0-\Pi_P)(1\mid W,A)\phi(W,A)(\Pi_P(1\mid W,A)\tau_{S,\beta_P}(W,A)-\Pi_0(1\mid W,A)\tau_{S,\beta_0}(W,A))\bigg]\\
&=P_0\bigg[(\Pi_0-\Pi_P)(1\mid W,A)\phi(W,A)(\Pi_P(1\mid W,A)\tau_{S,\beta_P}(W,A)-\Pi_0(1\mid W,A)\tau_{S,\beta_P}(W,A)\\
&+\Pi_0(1\mid W,A)\tau_{S,\beta_P}(W,A)-\Pi_0(1\mid W,A)\tau_{S,\beta_0}(W,A))\bigg]\\
&=-P_0\bigg[(\Pi_P-\Pi_0)^2(1\mid W,A)\phi(W,A)\tau_{S,\beta_P}(W,A)\bigg]\\
&+P_0\bigg[(\Pi_0-\Pi_P)(1\mid W,A)\phi(W,A)\Pi_0(1\mid W,A)(\tau_{S,\beta_P}-\tau_{S,\beta_0})(W,A)\bigg].
\end{align*}
Notice that $\mathbf{A}$ already has second-order term $(\Pi_P-\Pi_0)(\theta_P-\theta_0)$. The first term of $\mathbf{B}$ has second-order term $(\Pi_P-\Pi_0)(\tau_{S,\beta_P}-\tau_{S,\beta_0})$. Let $\mathbf{B}_{2}$ be the second term of $\mathbf{B}$, i.e.,
$$
\mathbf{B}_{2}\equiv P_0\left[\Pi_P(1-\Pi_P)(1\mid W,A)\phi(W,A)(\tau_{S,\beta_0}-\tau_{S,\beta_P})(W,A)\right].
$$
Since $I_{P}=E_P\Pi_P(1-\Pi_P)(1\mid W,A)\phi\phi^\top(W,A)$, we have
$$
\sum_jE_P\Pi_P(0\mid W,0)\phi_j(W,0)\left(I_{P}^{-1}\mathbf{B}_{2}\right)_j=P_0\left[\Pi_P(0\mid W,0)\tau_{S,\beta_0}(W,0)\right]-P_0\left[\Pi_P(0\mid W,0)\tau_{S,\beta_P}(W,0)\right].
$$
Note that $\mathbf{C}$ also has second-order terms $(\Pi_P-\Pi_0)^2$ and $(\Pi_P-\Pi_0)(\tau_{S,\beta_P}-\tau_{S,\beta_0})$. In addition, note that
\begin{align*}
&P_0\left[\Pi_P(0\mid W,0)\tau_{S,\beta_0}(W,0)\right]-P_0\left[\Pi_P(0\mid W,0)\tau_{S,\beta_P}(W,0)\right]\\
&-P_0\left[(\tau_{S,\beta_{P}}-\tau_{S,\beta_0})(W,0)(\Pi_P-\Pi_0)(0\mid W,0)\right]-P_0\left[\tau_{S,\beta_0}(W,0)(\Pi_P-\Pi_0)(0\mid W,0)\right]\\
&=P_0\left[\Pi_0(0\mid W,0)\tau_{S,\beta_0}(W,0)\right]-P_0\left[\Pi_P(0\mid W,0)\tau_{S,\beta_P}(W,0)\right]\\
&-P_0\left[(\tau_{S,\beta_P}-\tau_{S,\beta_0})(W,0)(\Pi_P-\Pi_0)(0\mid W,0)\right].
\end{align*}
Similarly,
\begin{align*}
&-P_0\left[\Pi_P(0\mid W,1)\tau_{S,\beta_0}(W,1)\right]+P_0\left[\Pi_P(0\mid W,1)\tau_{S,\beta_P}(W,1)\right]\\
&+P_0\left[(\tau_{S,\beta_P}-\tau_{S,\beta_0})(W,1)(\Pi_P-\Pi_0)(0\mid W,1)\right]+P_0\left[\tau_{S,\beta_0}(W,1)(\Pi_P-\Pi_0)(0\mid W,1)\right]\\
&=-P_0\left[\Pi_0(0\mid W,1)\tau_{S,\beta_0}(W,1)\right]+P_0\left[\Pi_P(0\mid W,1)\tau_{S,\beta_P}(W,1)\right]\\
&+P_0\left[(\tau_{S,\beta_P}-\tau_{S,\beta_0})(W,1)(\Pi_P-\Pi_0)(0\mid W,1)\right].
\end{align*}
To summarize, the exact remainder of $R^\#_{\mathcal{M}_{S,w}}(P,P_0)$ is
\begin{align*}
R^\#_{\mathcal{M}_{S,w}}(P,P_0)=P_0\bigg\{&\frac{g_P-g_0}{g_P}(1\mid W)\tau_{S,\beta_{P}}(W,1)(\Pi_P-\Pi_0)(1\mid W,1)\\
-&\frac{g_P-g_0}{g_P}(0\mid W)\tau_{S,\beta_P}(W,0)(\Pi_P-\Pi_0)(1\mid W,0)\\
-&(\tau_{S,\beta_P}-\tau_{S,\beta_0})(W,0)(\Pi_P-\Pi_0)(0\mid W,0)\\
+&(\tau_{S,\beta_P}-\tau_{S,\beta_0})(W,1)(\Pi_P-\Pi_0)(0\mid W,1)\\
+&\sum_jE_P[\Pi_P(0\mid W,0)\phi_j(W,0)-\Pi_P(0\mid W,1)\phi_j(W,1)]\bigg[I_{P}^{-1}\cdot\\
\big[&(\Pi_P-\Pi_0)(1\mid W,A)\phi(W,A)(\bar{Q}_P-\bar{Q}_0)(W,A)\\
+&(\Pi_P-\Pi_0)(1-\Pi_P)(1\mid W,A)\phi(W,A)(\tau_{S,\beta_P}-\tau_{S,\beta_0})(W,A)\\
-&(\Pi_P-\Pi_0)^2(1\mid W,A)\phi(W,A)\tau_{S,\beta_P}(W,A)\\
-&(\Pi_P-\Pi_0)(1\mid W,A)\phi(W,A)\Pi_0(1\mid W,A)(\tau_{S,\beta_P}-\tau_{S,\beta_0})(W,A)\big]\bigg]_j\bigg\}.
\end{align*}
In particular, note that if $\Pi_P=\Pi_0$, then the exact remainder $R^\#_{\mathcal{M}_{S,w}}(P,P_0)=0$, even if $\bar{Q}$ and $\tau_{S}$ are misspecified.
\end{proof}

\section{Details of the data-generating process for the oracle-bias evaluation}
\label{app:oracle_bias_dgp}

The numerical results in Table \ref{tab:oracle_bias} were obtained from a large-sample oracle calculation designed to assess the magnitude of the oracle bias induced by projecting the true conditional RCT-enrollment effect onto a simple working model. In this experiment, the true bias function was supplied directly, and we evaluated how much discrepancy remained after projecting this truth onto a misspecified working model. The sample size was set to $n=10^6$ so that Monte Carlo error would be negligible.

Let $\operatorname{expit}(x)=1/(1+\exp(-x))$. For each observation, we generated independent baseline variables and an outcome error $W_1,W_2,W_3,W_4,U_Y \sim \mathcal{N}(0,1)$. The working model used in the oracle projection included only $W_1,W_2,W_3$ and $A$. The fourth covariate, $W_4$, was generated but omitted from the working model. Therefore, in scenarios where the true bias depended on $W_4$, this represented an unobserved source of bias from the perspective of the working model. The RCT indicator was generated according to $S \sim \operatorname{Bern}\{p_S(W)\}$, where $p_S(W)=\operatorname{expit}(-0.5W_1-0.8)$. Treatment assignment was randomized in the RCT with probability $P_0(A=1\mid S=1,W)=g_{\rm rct}=0.5$, whereas treatment assignment in the external data depended on observed covariates: $P_0(A=1\mid S=0,W)=\operatorname{expit}(0.5W_1+0.3W_2)$. Thus, the treatment mechanism was known and randomized in the RCT, but confounded by measured covariates in the external data.

The true treatment effect was set to $0.5$. The oracle-bias calculation focused on the conditional average RCT-enrollment effect
$$
\tau_{S,0}(W,A)=E_0(Y\mid S=1,W,A)-E_0(Y\mid S=0,W,A),
$$
which is the bias function appearing in the bias estimand $\Psi^\#(P_0)$. We considered three classes of bias functions, indexed by $\rho\in\{0.1,0.2,\ldots,1.0\}$:
\begin{align*}
\text{Unobserved:}\qquad
\tau_{S,0}(W,A)
&=-(0.5+\rho W_4)(1-A),\\
\text{Misspecified:}\qquad
\tau_{S,0}(W,A)
&=-(0.2+0.1\rho W_2^3)(1-A),\\
\text{Unobserved + misspecified:}\qquad
\tau_{S,0}(W,A)
&=-(0.5+0.1\rho W_1^3+\rho W_4)(1-A).
\end{align*}
The first setting creates bias through the omitted variable $W_4$. The second setting creates bias through a nonlinear function of an observed covariate, while the oracle working model contains only linear main terms. The third setting combines both forms of misspecification. Larger values of $\rho$ correspond to increasingly complex or increasingly difficult-to-approximate bias functions. The rows of Table \ref{tab:oracle_bias} correspond, from top to bottom, to $\rho=0.1,\ldots,1.0$.

For each simulated data set, we fit the oracle weighted least-squares projection
$$
\widehat{\tau}_{S,\mathrm{or}}(W,A)=\widehat{\beta}_0+\widehat{\beta}_1W_1+\widehat{\beta}_2W_2+\widehat{\beta}_3W_3+\widehat{\beta}_4A,
$$
using the known true bias function $\tau_{S,0}(W,A)$ as the outcome. The projection was weighted by $\pi_1(W)\{1-\pi_1(W)\}$, where
$$
\Pi_0(1\mid W,A)=\frac{P_0(A\mid S=1,W)p_S(W)}{P_0(A\mid S=1,W)p_S(W)+P_0(A\mid S=0,W)\{1-p_S(W)\}}.
$$
This projection represents an oracle version of the bias working model. It is given direct access to the true $\tau_{S,0}$, but is still constrained to use the same low-dimensional linear working model. Therefore, any remaining discrepancy reflects approximation error due to the working model rather than estimation error in $\tau_{S,0}$.

For each scenario, we reported the adjusted $R^2$ of this oracle weighted regression, the weighted mean squared error
$$
\frac{1}{n}\sum_{i=1}^n\pi_1(W_i)\{1-\pi_1(W_i)\}\left[\widehat{\tau}_{S,\mathrm{or}}(W_i,A_i)-\tau_{S,0}(W_i,A_i)\right]^2,
$$
and the empirical oracle-bias measure
$$
\frac{1}{n}\sum_{i=1}^n\{S_i-\pi_1(W_i)\}\left[\left\{\widehat{\tau}_{S,\mathrm{or}}(W_i,1)-\tau_{S,0}(W_i,1)\right\}-\left\{\widehat{\tau}_{S,\mathrm{or}}(W_i,0)-\tau_{S,0}(W_i,0)\right\}\right].
$$
This final quantity measures the residual contribution of the projection error to the treatment-contrast scale after the true bias function is approximated by the oracle working model. Across the scenarios in Table \ref{tab:oracle_bias}, the oracle bias remained small relative to the true treatment effect of $0.5$, even when the adjusted $R^2$ decreased substantially as $\rho$ increased.

\section{Data-generating processes and additional simulations}\label{app:additional_sims}

In the simulations described in Section \ref{sec:simulations}, we generated RCT data and external RWD data from five different sources, with each source exhibiting varying degrees of covariate shift relative to the RCT covariate distribution. For the RCT, three baseline covariates, $W_1$, $W_2$, and $W_3$, were drawn from standard normal distributions. For the external RWD from source $j$, where $j\in\{1,2,3,4,5\}$, the covariate shift was introduced by generating $W_1$ and $W_3$ from $\mathcal{N}(0.2(j-1),1)$ and $W_2$ from $\mathcal{N}(-0.2(j-1),1)$. The treatment variable $A$ for the RCT was sampled from a binomial distribution with a probability of 0.5. For the external RWD, $A$ was drawn from a binomial distribution with probability $\text{expit}(-2+1.6W_1-2W_2)$. The outcome is defined as: $Y=2.5+0.9W_1+1.1W_2+2.7W_3+0.5A+U_Y+(1-S)b$, where $U_Y\sim\mathcal{N}(0,3)$, and the bias term $b$ was defined as:
\begin{align*}
b&= \left\{
\begin{array}{ll}
0 & \text{if } j=1; \\
0.5+1.4W_1A & \text{if } j\in\{2,3\}; \\
0.5+1.4W_1A+0.2W_3 & \text{if } j=4;\\
1.3+1.4W_1A+0.2W_3 & \text{if } j=5.
\end{array}
\right.\\
\end{align*}
\begin{table}[h!]
\centering
\resizebox{\columnwidth}{!}{
\begin{tabular}{ccccccc}
\toprule
Sample size of external arms & Strategy & $|\text{Bias}|$ & Variance & 95\% CI widths & 95\% CI coverage & Folds pooled (\%) \\
\midrule
\multirow{1}{*}{0} & RCT data only & 0.002 & 0.081 & 1.169 & 0.97 & - \\
\midrule
\multirow{2}{*}{500} 
& Random & 0.241 & 0.157 & 1.166 & 0.75 & 0.358 \\
& TES + PS matching & \textbf{0.077} & \textbf{0.075} & \textbf{1.067} & \textbf{0.95} & \textbf{0.588} \\
\midrule
\multirow{2}{*}{600} 
& Random & 0.289 & 0.203 & 1.179 & 0.72 & 0.390 \\
& TES + PS matching & \textbf{0.080} & \textbf{0.072} & \textbf{1.068} & \textbf{0.95} & \textbf{0.566} \\
\midrule
\multirow{2}{*}{700} 
& Random & 0.284 & 0.186 & 1.175 & 0.68 & 0.376 \\
& TES + PS matching & \textbf{0.089} & \textbf{0.078} & \textbf{1.081} & \textbf{0.95} & \textbf{0.514} \\
\midrule
\multirow{2}{*}{800} 
& Random & 0.270 & 0.174 & 1.177 & 0.73 & 0.354 \\
& TES + PS matching & \textbf{0.100} & \textbf{0.077} & \textbf{1.084} & \textbf{0.95} & \textbf{0.502} \\
\midrule
\multirow{2}{*}{900} 
& Random & 0.281 & 0.180 & 1.182 & 0.73 & 0.366 \\
& TES + PS matching & \textbf{0.104} & \textbf{0.075} & \textbf{1.084} & \textbf{0.95} & \textbf{0.484} \\
\midrule
\multirow{2}{*}{1000} 
& Random & 0.251 & 0.167 & 1.188 & 0.75 & 0.312 \\
& TES + PS matching & \textbf{0.103} & \textbf{0.077} & \textbf{1.096} & \textbf{0.95} & \textbf{0.460} \\
\bottomrule
\end{tabular}
}
\caption{Absolute bias, variance, average 95\% confidence interval widths, coverage, and proportion of folds that selected the pooled-ATE estimand averaged across simulation runs by external data sample size and sampling strategy. ``RCT data only" denotes the RCT-only design, analyzed with the TMLE estimator. ``Random" denotes a randomly sampled external cohort of the specified size, analyzed with the ES-CVTMLE estimator. ``TES + PS matching" denotes our proposed trial enrollment score and propensity score matching strategy for sampling the external cohort of the specified size, also analyzed with the ES-CVTMLE estimator.}
\label{tab:sim_escvtmle}
\end{table}

We also tried applying an alternative data integration estimator, ES-CVTMLE \citep{dang_escvtmle_2022}, under the same simulation setup. The results are consistent with those of A-TMLE shown in Section \ref{sec:simulations}. Specifically, as shown in Table \ref{tab:sim_escvtmle}, in challenging scenarios where ES-CVTMLE gives low confidence interval coverage under the random sampling regime due to poor estimation of the bias $\Psi^\#$, our proposed strategy brings coverage back to nominal. The ES-CVTMLE method works by splitting the data into $V$ folds. For each fold, it first determines whether to estimate the pooled-ATE estimand $\tilde{\Psi}(P_0)$ or the covariate-pooled ATE estimand $\Psi(P_0)$. This decision is based on an experiment selection criterion that optimizes the bias-variance trade-off for the target parameter, using the larger training set. The chosen target parameter is then estimated on the validation fold. This process is repeated across all $V$ folds, and the final estimate is obtained as the average across folds. The expected efficiency gain can also be inferred from the proportion of folds selecting the pooled-ATE estimand. Thus, in the last column of Table \ref{tab:sim_escvtmle}, we report the average proportion of folds that selected the pooled-ATE estimand across simulation runs. Our results indicate that the proposed matching strategy increases the proportion of folds pooling external data, suggesting that the matching helps reduce the bias in the external data. Additionally, we observe that for ES-CVTMLE, the number of folds pooling external data decreases as the sample size of the external data grows. This is expected, since as more external data is used for data integration, the estimated bias becomes larger than the standard error of the estimator so ES-CVTMLE tends to select the pooled-ATE estimand less often.

\end{document}